\documentclass{article}

\usepackage{fullpage}
\usepackage{authblk}

\usepackage{tikz-cd}
\usetikzlibrary {arrows.meta} 
\newlength{\ptsize}
\usepackage{graphicx,amsmath,amssymb,amsfonts,amsthm,mathtools,hyperref,cleveref} 
\usepackage[dvipsnames]{xcolor}
\usepackage[numbers,sort]{natbib}
\usepackage{csquotes}
\usepackage{booktabs}
\usepackage{enumerate}
\usepackage{pdfpages}
\usepackage[normalem]{ulem}

\definecolor{linkcolor}{rgb}{0,0.2,0.6}
\hypersetup{colorlinks,breaklinks,urlcolor=linkcolor,linkcolor=linkcolor,citecolor=linkcolor}

\newcommand{\cS}{\mathcal S}
\newcommand{\cL}{\mathcal L}

\DeclareMathOperator*{\argmax}{arg\,max}

\newtheorem{theorem}{Theorem}[section]

\newtheorem{prop}[theorem]{Proposition}
\newtheorem{conjecture}{Conjecture}[section]

\theoremstyle{definition}

\newtheorem{definition}[theorem]{Definition}

\newtheorem{obs}[theorem]{Observation}

\hypersetup{
    colorlinks=true,
    linkcolor=blue,
}

\author[1,2,5]{Laura Kubatko}
\author[3]{Simone Linz}
\author[4,5]{Kristina Wicke}

\affil[1]{Department of Evolution, Ecology, and Organismal Biology, The Ohio State University, Columbus, OH, USA}
\affil[2]{Department of Statistics, The Ohio State University, Columbus, OH, USA}
\affil[3]{School of Computer Science, University of Auckland, New Zealand}
\affil[4]{Department of Mathematical Sciences, New Jersey Institute of Technology, Newark, NJ, USA}
\affil[5]{National Institute for Theory and Mathematics in Biology, Northwestern University and The University of Chicago, Chicago, IL, USA}

\title{Revisiting a random model of lateral gene transfer in phylogenetics}
\date{}

\begin{document}
\maketitle

\begin{abstract}
Processes such as incomplete lineage sorting (ILS) and reticulate evolution (arising, for example, from hybridization and lateral gene transfer (LGT)) are known to cause discordance between gene trees and species trees, complicating phylogenetic inference. While the multispecies coalescent model has led to a rich theoretical understanding of ILS, probabilistic models for LGT have received comparatively less attention. 
Here, we revisit a basic LGT model in which random LGT events occur according to a Poisson process with a constant transfer rate. Focusing on the simplest cases of two and three species, we derive gene tree and site pattern probabilities and discuss their implications for model identifiability. We also address the question of whether LGT and ILS can be distinguished from one another under these probabilistic models. We discuss empirical applications and  outline directions for future work. 
\end{abstract}


\section{Introduction}
A central challenge in phylogenetics is the estimation of a species tree from genomic data~\cite{kubatko2026evolving,kubatko2023species}.  Due to recent advances in sequencing technologies,  whole genome sequence data is now commonly available and offers an opportunity for species tree inference that estimates a sequence of speciation events while integrating evolutionary processes such as incomplete lineage sorting (ILS), hybridization, and lateral gene transfer (LGT), which are known to cause  discordance between a species tree and individual gene trees. 

The multispecies coalescent (MSC) model facilitates species tree inference in the presence of ILS~\cite{degnan2005gene,rannala2003bayes,degnan2009gene}. 
In addition to modeling ILS, the MSC model can also be used to tackle other biological problems, including the estimation of species divergence times and population sizes of ancestral species. Recently, the MSC model has been extended from species trees to species networks to model gene flow (e.g., hybridization)~\cite{degnan2018modeling,yu2011coalescent}. In contrast, frameworks to model LGT remain poorly understood, although LGT is an important evolutionary process that is prevalent  in many groups of organisms, ranging from bacteria to eukaryotes~\cite{sieber2017lateral,soucy2015horizontal}.  Organisms acquire entirely new genes by means of LGT and, as a result, are frequently able to adapt to new environments. From a medical perspective, LGT is also a key player in rapidly spreading  antibiotic resistance~\cite{liu2022antimicrobial}. 

In~\cite{linz2007}, the authors describe a stochastic model of LGT in which transfer events are modeled as subtree prune and regraft operations on a species tree so that each LGT event occurs between two contemporaneous lineages uniformly at random. Moreover, within a given time interval, the number of  transfers events follows a Poisson process.  The model was introduced in 2007 with the intent to estimate a rate of LGT from a collection of gene trees in a likelihood framework. Since then several theoretical questions that arise in the study of LGT have been addressed using the same LGT model or one of the closely related models (e.g.,~\cite{galtier2007,roch2013,suchard2005stochastic}). For example, in the context of rooted species trees and unrooted gene trees, Roch and Snir~\cite{roch2013} asked how much LGT can be handled before the tree-like signal is lost. To answer this question, they showed that a  species tree can be estimated accurately from a number of gene trees that is logarithmic in the number of considered extant species provided that the expected number of LGT events falls below a certain threshold. More advanced results on the threshold were subsequently published in~\cite{Daskalakis2015}. Complementary to these results, Steel et al.~\cite{steel2013} focused on rooted species and gene trees and established a sufficient condition for statistically consistent species tree estimation for an arbitrary number of species. Loosely speaking,  this condition is an upper bound on the expected number of LGT events of a particular type. This last result also holds under an extended LGT model in which transfer events between two species are not uniformly distributed, but  instead occur with a probability that depends on the distance between them so that LGT events become less likely for more distantly related species. 

In this paper, we revisit the aforementioned LGT model and analyze it from several additional angles. While previous studies have considered gene tree probabilities under the LGT model~\cite{steel2013,steel2016}, we analyze site patterns, which are the nucleotides observed at the leaves of a species tree. Specifically, we  consider the Jukes-Cantor model of DNA sequence evolution and derive site pattern probabilities under the LGT model for species trees with two and three leaves. Enhancing our understanding of site patterns is crucial for the development of software to infer species trees in the presence of LGT that make direct use of DNA sequences without having to estimate gene trees first. Second, we investigate which parameters of the  LGT model  are identifiable from site pattern or gene tree distributions. 
Here, the parameters of interest are the transfer rate and the coalescence times. If the model parameters are identifiable, then this implies that they can be estimated accurately from sufficiently large (empirical) datasets such as observed site patterns. We defer a precise definition of different types of identifiability to a later section. In addition to establishing theoretical identifiability results, we also devise a maximum likelihood framework that estimates the LGT model parameters from observed site patterns in the 3-taxon case and evaluate its performance via simulations. 

The remainder of the paper is organized as follows. The next section provides mathematical definitions and concepts that are used throughout the following sections. Subsequently, in Sections~\ref{sec:two} and~\ref{sec:three}, we model LGT for rooted species and gene trees with two and three taxa respectively, derive gene tree probabilities and site pattern probabilities under the Jukes-Cantor model, and establish identifiability or non-identifiability of the LGT model parameters and the species tree topology  from gene tree and site pattern probabilities. In Section~\ref{sec:estimate}, we show that maximum likelihood can be used to reliably estimate the parameters of the LGT model in the 3-taxon case. Lastly, in Section~\ref{sec:ILS}, we compare the gene tree probabilities under the MSC and LGT model and explore parameters for which the MSC and LGT model are indistinguishable.

\section{Preliminaries} 
To formally state our results, we need some terminology and notation. We mainly follow the notation of \citet{steel2013}.
Throughout this paper, $X$ denotes a non-empty finite set (of taxa or species). If not stated otherwise, $X = \{1, 2, \ldots, n\}$. 

\paragraph{Species trees and related concepts.}
A {\em rooted binary species tree} $\cS = (V,E)$ with leaf set $X$ is a rooted tree that satisfies the following three properties: (i) the unique root $\rho$ has in-degree zero and out-degree two, (ii) the leaves are bijectively labeled with the elements in $X$, and (iii) each remaining vertex has in-degree one and out-degree two. Since all species trees in this paper are rooted and binary, we will refer to a rooted binary species tree simply as a \emph{species tree}.

Two leaves $x$ and $y$ are called a {\em cherry} if there exists a vertex $u$ such that $(u,x)$ and $(u,y)$ are arcs.

In the following, we regard $\cS$ as a 1-dimensional simplicial complex, so each ``point" $p$ in $\cS$ is either a vertex or an element of the interval that corresponds to an arc.
Further, we consider a coalescence time scale $t: \cS \rightarrow [0, \infty)$ of the tree with coalescence time increasing into the past. Then,
    \begin{itemize}
        \item $t(p) = 0$ if and only if $p$ is a leaf;
        \item for two distinct points $u$ and $v$, if $u$ lies on a direct path from $v$ to a leaf, then $t(u) < t(v)$.
    \end{itemize}
We refer to $t(p)$ as the $t$-value of $p$.
Notice that we view $\cS$ as a {\em clocklike} or {\em ultrametric} tree with arc lengths induced by the coalescence time scale $t$. We use $\cL(\cS)$ to denote the total sum of arc lengths (also known as {\it phylogenetic diversity}) of $\cS$. Note that although we refer to times at which collections of taxa share an ancestor as ``coalescence times,''
in what follows, we model only the LGT process (i.e., our LGT model does not include ILS as modeled  by the MSC, though we do briefly discuss the MSC in Section \ref{sec:ILS}).

\paragraph{The Jukes-Cantor model and site pattern probabilities.} 
The Jukes-Cantor (JC69) model \cite{jukescantor1969} is a basic Markov model of DNA sequence evolution. More precisely, it is a continuous-time, time-reversible substitution model on four states $(A, C, G, T)$, in which all substitutions occur at the same rate. Along a branch of length $t$, the probability of no change is
\[p_{\mathrm{same}}(t) = \frac{1}{4} + \frac{3}{4}e^{-\mu t},\]
while the probability of a change to any specific alternative nucleotide is
\[p_{\mathrm{diff}}(t) = \frac{1}{4} - \frac{1}{4}e^{-\mu t}.\]
Thus, conditional on the ancestral state, all three possible substitutions occur with equal probability. In what follows, we adopt the common normalization $\mu = 4/3$, which simplifies subsequent expressions.

A \emph{site pattern} refers to the observed configuration of nucleotides at a given site across the taxa under consideration. For two taxa, a site pattern is an ordered pair of states in $\{A, C, G, T\}^2$, describing the nucleotides observed in each taxon at that site. Under the JC69 model, site pattern probabilities depend only on whether the states are the same or different, rather than on the specific nucleotide labels. In particular, all matching patterns have equal probability (e.g., $p_{AA} = p_{CC} = \cdots$), and all mismatching patterns have equal probability. We therefore summarize site patterns by $p_{xx}$ for matches and $p_{xy}$ for mismatches in the 2-taxon case (and analogously in the 3-taxon case; see Appendix~\ref{app:site-pattern-probabilities}).

\paragraph{Lateral gene transfer events, transfer sequences, and associated gene trees.}
A {\em lateral gene transfer (LGT) event} (or {\em transfer event} for short) on $\cS$ is an arc $(p,p')$ from  $p \in \cS$ to $p' \in \cS$ where $p$ and $p'$ are contemporaneous, i.e., $t(p) = t(p')$. We will assume that neither $p$ nor $p'$ are vertices of $\cS$, implying that transfers occur between points on the arcs of $\cS$. We also refer to $(p,p')$ as a {\em transfer arc}.

We write $\sigma = (p, p')$ to denote this transfer event and we write $t(\sigma)$ for the common value of $t(p)$ and $t(p')$. Further, we will assume that no two transfer events occur at exactly the same time.

Let $\underline{\sigma} = (\sigma_1, \ldots, \sigma_k)$ be a sequence of transfer events in increasing $t$-value order  with $\sigma_i = (p_i, p_i')$ for each $i\in\{1,\ldots,k\}$. By the definition of $t$-values, $\sigma_1$ is the most recent transfer event and $\sigma_k$ is the most ancient. Since we assume that no two transfer events occur at the same time, this ordering is well-defined, and we have 
\[ 0 < t(\sigma_1) < t(\sigma_2) < \ldots < t(\sigma_k) < t(\rho).\]
We refer to $\underline{\sigma}$ as a {\em transfer sequence}. If no transfer events occur on $\cS$, then \underline{$\sigma$} is the empty sequence. \smallskip

Given a species tree $\cS$ and a transfer sequence $\underline{\sigma} = (\sigma_1, \ldots, \sigma_k)$ on $\cS$, we obtain an associated {\em gene history}  $T[\underline{\sigma}]$ in the following way: We assume that a transfer arc $(p,p')$ replaces the gene that was present on the arc at $p'$ with the transferred gene from $p$.
Thus, tracing the history of the gene backwards in time (from the present to the past), each time we encounter an incoming transfer arc into an arc of $\cS$, we follow this arc against its direction to obtain $T[\underline{\sigma}]$.

More formally, let $\underline{\sigma} = (\sigma_1, \ldots, \sigma_k)$ be a transfer sequence on a species tree $\cS$, where $\sigma_i = (p_i, p_i')$. For $r \in \{0, \ldots, k\}$, let $T_r$ denote the simplicial complex obtained from $\cS$ by adding only the first $r$ transfer arcs $(p_i,p_i')$ with $i \in \{1,\ldots, r\}$, and deleting the interval immediately above each target point $p_i'$ with $i \in \{1, \ldots, r\}$. Thus, $T_0 = \mathcal{S}$, and $T_k$ is the simplicial complex associated with the full transfer sequence. Finally, $T[\underline{\sigma}]$ is the minimal connected subgraph of $T_k$ that contains $X$. We illustrate these concepts in Fig.~\ref{fig:transfer}.

We say that a transfer $\sigma_j = (p_j, p_j')$ is \emph{relevant} if both $p_j$ and $p_j'$ are points of $T_{j-1}$. Equivalently, immediately before applying $\sigma_j$, both its source and target are present in the simplicial complex $T_{j-1}$ obtained from applying the first $j-1$ transfers. For example, in Fig.~\ref{fig:transfer}, the transfers $\sigma_1$ and $\sigma_3$ are relevant, whereas $\sigma_2$ is not.

\begin{figure}[t!]
    \centering
    \begin{tikzpicture}[thick,scale=1]

    \draw [|-{Latex[length=2.5mm]}](0.5,0) -- (0.5,4);
    \node[align=left] at (0.5,-0.5) {time $t$};
    
        \node[fill=black,circle,inner sep=1.5pt, label=below: {$1$} ] at (2,0){};
        \node[fill=black,circle,inner sep=1.5pt, label=below: {$2$} ] at (3,0){};
        \node[fill=black,circle,inner sep=1.5pt, label=below: {$3$} ] at (5,0){};
        \node[fill=black,circle,inner sep=1.5pt, label=left: {$u$}]  at (2.5,1){};
        \node[fill=black,circle,inner sep=1.5pt, label=above: {$\rho$}] at (3.5,3){};
        \draw(3.5,3)--(2,0);
        \draw(3.5,3)--(5,0);
        \draw(2.5,1)--(3,0);
        \draw[-{Latex[length=2.5mm]},dashed](4.75,0.5)--(2.75,0.5) node[midway,below] {$\sigma_1$};
        \draw[-{Latex[length=2.5mm]},dashed](4.625,0.75)--(2.625,0.75) node[midway,above] {$\sigma_2$};
        \draw[-{Latex[length=2.5mm]},dashed](3,2)--(4,2) node[midway,below] {$\sigma_3$};
        \node[align=left] at (2,3) {\large $\cS$};
        \node[align=left] at (3.5,-1.5) {\large (i)};

        \node[fill=black,circle,inner sep=1.5pt, label=below: {$1$} ] at (7,0){};
        \node[fill=black,circle,inner sep=1.5pt, label=below: {$2$} ] at (8,0){};
        \node[fill=black,circle,inner sep=1.5pt, label=below: {$3$} ] at (10,0){};
        \draw(8.5,3)--(7,0);
        \draw(9,2)--(10,0);
        \draw(7.75,0.5)--(8,0);
        \node[fill=black,circle,inner sep=1.5pt]  at (7.5,1){};
        \node[fill=black,circle,inner sep=1.5pt] at (8.5,3){};
        \draw[-{Latex[length=2.5mm]},dashed](9.75,0.5)--(7.75,0.5) node[midway,below] {$\sigma_1$};
        \draw[-{Latex[length=2.5mm]},dashed](9.625,0.75)--(7.625,0.75) node[midway,above] {$\sigma_2$};
        \draw[-{Latex[length=2.5mm]},dashed](8,2)--(9,2) node[midway,below] {$\sigma_3$};
        \node[align=left] at (8.5,-1.5) {\large (ii)};

        \node[fill=black,circle,inner sep=1.5pt, label=below: {$1$} ] at (12,0){};
        \node[fill=black,circle,inner sep=1.5pt, label=below: {$2$} ] at (13,0){};
        \node[fill=black,circle,inner sep=1.5pt, label=below: {$3$} ] at (15,0){};
        \node[fill=black,circle,inner sep=1.5pt]  at (13,2){};
        \node[fill=black,circle,inner sep=1.5pt] at (14.5,0.5){};
        \draw(13,2)--(12,0);
        \draw(13,2)--(15,0);
        \draw(14.5,0.5)--(13,0);
        \node[align=left] at (12,3) {\large $T[\underline{\sigma}]$};
        \node[align=left] at (13.5,-1.5) {\large (iii)};
     \end{tikzpicture}
  \caption{(i) A species tree $\cS$ with a sequence $\underline{\sigma}$ of three transfer events, labeled in increasing order into the past; (ii) the simplicial complex obtained from $\cS$ by deleting the interval immediately above the endpoint of each transfer;  (iii) the resulting tree $T[\underline{\sigma}]$. Notice that $\cS$ and $T[\underline{\sigma}]$ depict different evolutionary relationships.}
    \label{fig:transfer}
\end{figure}
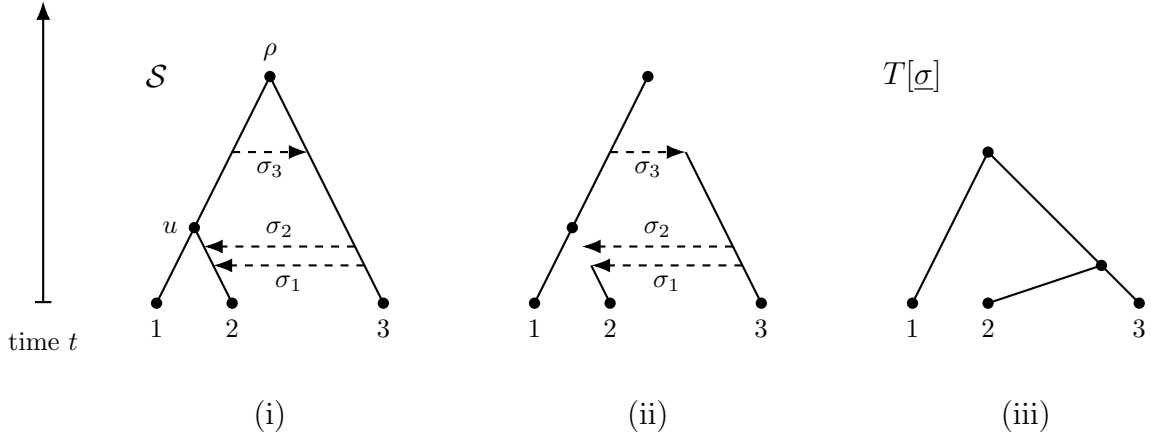

\paragraph{The standard LGT model.}
We now recall the {\em standard LGT model} as introduced by~\citet{linz2007}. Their model, in which LGT events within a given time interval follow a Poisson process, is based on the following six assumptions: (i) the species tree $\mathcal{S}$ is rooted, binary, and clocklike, (ii) all incongruence between $\mathcal{S}$ and a gene tree is caused by LGT events, (iii) the transfer rate, say $\lambda$, is constant across $\mathcal{S}$, (iv) genes are transferred independently, (v) the donor species keeps a copy of the transferred gene, and (vi) the transferred gene replaces any existing orthologous gene in the recipient species.

Notice that under this model, the number of transfers has a Poisson distribution, with mean equal to the transfer rate times the sum of branch lengths of the tree. In other words, the number of transfers has a Poisson distribution with parameter 
$\lambda \cdot \cL(\cS)$. 
In addition, following from the assumption of a Poisson process, the time until the next transfer event follows an exponential distribution with mean $\frac{1}{\lambda}$ (so that the mean arrival rate per unit time is $\lambda$).

\paragraph{(Generic) identifiability.}
In the following sections, we study which parameters of the standard LGT model are identifiable from gene tree probabilities and from site pattern probabilities. Let $\theta$ denote the parameter vector of a statistical model, and let $f(\theta)$ be a function of these parameters (for example, the induced gene tree probability distribution in case of the standard LGT model). We say that \emph{$\theta$ is identifiable from $f$} if $f(\theta) = f(\theta')$ implies $\theta = \theta'$, i.e., $f$ is injective on the parameter space.

A parameter is said to be \emph{generically identifiable from $f$} if this injectivity holds for all parameter values except possibly on a subset of measure zero (equivalently, in algebraic settings, on a proper lower-dimensional subset) of the parameter space.

\section{The 2-taxon case}\label{sec:two}
We begin with the standard LGT model in its simplest setting of two taxa. Consider the subtree of $\mathcal{S}$ (Fig.~\ref{fig:transfer}) induced by taxa 1 and 2 and their common parent $u$, and denote it by $\mathcal{T}$. Let $e_1$ and $e_2$ denote the arcs incident to taxa 1 and 2, respectively, and assume that the $t$-value of $u$ is $t_1 > 0$. Let $\underline{\sigma} = (\sigma_1, \ldots, \sigma_k)$ be a transfer sequence.

Under the standard LGT model, there can be at most one relevant transfer, denoted $\sigma^\ast = \sigma_1$. Indeed, exactly one of the following holds: (i) no relevant transfer occurs during the interval $(0, t_1)$, or (ii) a single relevant transfer occurs in $(0, t_1)$ between arcs $e_1$ and $e_2$. In either case, the resulting gene history $T[\underline{\sigma}]$ has the same topology as the species tree, although its arc lengths differ between the two scenarios.

\begin{enumerate}
    \item \textbf{No relevant transfer.} \\
    If no transfer occurs, then $T[\underline{\sigma}]$, denoted by $H_0$, has root $t$-value equal to $t_1$. 
    
    \item \textbf{Exactly one relevant transfer in $(0,t_1)$.} \\
    Suppose the transfer $\sigma^\ast$ occurs at time $h_1 \in (0, t_1)$. In this case, $T[\underline{\sigma}]$, denoted by $H_1$, has root $t$-value equal to $h_1$.
\end{enumerate}

\paragraph{Density of the gene divergence time.}
We next derive the density of the time of gene divergence. Note that for $H_0$, corresponding to the case of no relevant transfers, the divergence time is always $t_1$ and thus we have a point mass at $t_1$ with probability $e^{-2 \lambda t_1}$. For $H_1$, suppose that the unique relevant transfer occurs at time $h \in (0,t_1)$.  This follows an exponential density with rate $2 \lambda$. In both this and the previous expression, the factor ``2'' arises from the two branches under consideration. Combining these, we have the following density for the divergence time, $h$:
\begin{align} \label{eq:fulldensity:2taxa}
    f(h | \mathcal{T}, t_1, \lambda) = 
    \begin{cases}
    2\lambda \, e^{-2\lambda h}, & 0 < h < t_1 \quad  \mbox{(gene history }H_1)\\
 e^{-2 \lambda t_1}, & h=t_1 \quad \quad \quad  \mbox{(gene history }H_0) \\
      0, & \mbox{otherwise}\\
    \end{cases}
\end{align}
One can check that $f(h | \mathcal{T}, t_1, \lambda)$ is a valid density for the divergence time $h$.

\paragraph{Site pattern probabilities under the JC69 model.}
We next derive site pattern probabilities under the JC69 model for the standard LGT model with two taxa.

Consider the gene history $H_0$. Recall that $t_1$ denote the $t$-value at its root. Define
\begin{align*}
   a_0 &= \frac{1}{4} + \frac{3}{4} e^{-\mu t_1}, \qquad
   a_1 = \frac{1}{4} - \frac{1}{4} e^{-\mu t_1}.
\end{align*}
Then, the probabilities of the site patterns $xx$ and $xy$ given $H_0$, respectively, are
\begin{align*}
    \mathbb{P}(xx | H_1) &= a_0^2 + 3 a_1^2, \\
    \mathbb{P}(xy | H_1) &= 6 a_0 a_1 + 6 a_1^2.
\end{align*}
For $H_1$, analogous expressions hold, with $t_1$ replaced by $h_1$ in the definitions of $a_0$ and $a_1$.

By averaging the conditional site pattern probabilities over the LGT-induced distribution of gene histories and transfer times (i.e., over $H_0$ and $H_1$ and their associated times), we obtain the marginal site pattern probabilities $\boldsymbol{\overline{p}} = (p_{xx}, p_{xy})$. Using \texttt{Mathematica}~\cite{mathematica} and setting $\mu = 4/3$, we obtain: 

\begin{align}
    p_{xx} &=  \mathbb{P}(xx | H_0) \cdot e^{-2 \lambda t_1} + \int\limits_{h=0}^{t_1} \mathbb{P}(xx | H_1) \cdot 2\lambda \, e^{-2\lambda h} \, dh = \frac{3 \lambda +3 e^{-\frac{2}{3} (3 \lambda +4) t_1}+1}{3 \lambda +4}; \label{pxx}\\
    p_{xy} &= \mathbb{P}(xy | H_0) \cdot e^{-2 \lambda t_1} + \int\limits_{h=0}^{t_1} \mathbb{P}(xy | H_1) \cdot 2\lambda \, e^{-2\lambda h} \, dh = \frac{3-3 e^{-\frac{2}{3}  (3 \lambda +4) t_1}}{3 \lambda +4}. \label{pxy}
\end{align}

\subsection{Identifiability results}
We now consider the question of which parameters are identifiable under the standard LGT model.  We consider two cases: identifiability when the gene history probabilities are known and identifiability from the distribution of site patterns. Note that in practice the gene history probability distribution cannot be observed, nor can it be estimated since error associated with estimated coalescence times would prohibit classification of estimated gene histories into the two cases $H_0$ and $H_1$. However, we consider this case in Section \ref{subsubsec:genehistories.2tax} for illustrative purposes, as it will be helpful for understanding the situation in the case of three taxa for which the gene tree distribution can be estimated.

\subsubsection{Identifiability of parameters from gene history probabilities}\label{subsubsec:genehistories.2tax}
As discussed above, two possible gene histories can occur: $H_0$, which corresponds to the case of no relevant transfer, and $H_1$, which corresponds to the case in which a transfer event occurs in the interval $(0, t_1)$. These events occur with probabilities $e^{-2 \lambda t_1}$ and $1-e^{-2 \lambda t_1}$, respectively. Noting that these probabilities sum to 1, it is clear that only one parameter can be estimated since there is only one degree of freedom.  We thus show that $t_1$ is generically identifiable given $\lambda$, and conversely, that $\lambda$ is generically identifiable given $t_1$.

\begin{prop}\label{prop:2taxa:t-given-lambda}
Fix $\lambda > 0$, and let $\mathbb{P}_{H_0}$ be the probability of no relevant transfers.  Then the root time $t_1$ is given by 
\begin{equation}\label{eq:2tax-ident-t1}
    t_1 = \frac{- \log (\mathbb{P}_{H_0})}{2 \lambda}.    
\end{equation}
In particular, for fixed $\lambda > 0$, the parameter $t_1$ is generically identifiable from the gene history probabilities. \\

\noindent Conversely, suppose that $t_1>0$ is fixed. Then $\lambda$ is given by 
\begin{equation}\label{eq:2tax-ident-lambda}
    \lambda = \frac{- \log (\mathbb{P}_{H_0})}{2 t_1}.    
\end{equation}
Thus, for fixed $t_1 > 0$, the parameter $\lambda$ is generically identifiable from the gene history probabilities.
\end{prop}

\begin{proof}
Noting that $\mathbb{P}_{H_0} = e^{-2 \lambda t_1}$, 
Eqns.~\eqref{eq:2tax-ident-t1} and \eqref{eq:2tax-ident-lambda} follow from solving for $t_1$ and $\lambda$, respectively.
\end{proof}

\subsubsection{Identifiability of parameters from site pattern probabilities}
We next turn to the question of which parameters of the standard LGT model are identifiable from the site pattern probability distribution under the JC69 model in the 2-taxon setting.

\medskip
Since $p_{xx} + p_{xy} = 1$, the site pattern distribution has only one degree of freedom. It follows immediately that $\lambda$ and $t_1$ cannot be identified simultaneously from the distribution. However, we next show that $t_1$ is generically identifiable given $\lambda$, and conversely, $\lambda$ is generically identifiable given $t_1$.

\begin{prop}\label{prop:3taxa:t-given-lambda}
Fix $\lambda > 0$. Then the root time $t_1$ is given by
\begin{align*}
    t_1 &= \frac{\ln \left(p_{xx} - \frac{3\lambda+1}{3} p_{xy}\right)}{- \frac{2}{3} (3+4\lambda)}.
\end{align*}
In particular, for fixed $\lambda > 0$, the parameter $t_1$ is generically identifiable from the site pattern probability distribution.
\end{prop}

\begin{proof}
Define the composite parameter\[a = e^{-\frac{2}{3} (3\lambda + 4)t_1}.\]
From Eqns.~\eqref{pxx} and \eqref{pxy}, we have
\begin{align*}
p_{xx} &= \frac{3\lambda + 3a +1}{3\lambda+4},
\qquad
p_{xy} = \frac{3-3a}{3\lambda+4}.
\end{align*}
Solving the second equation for $a$, we obtain
\begin{align*}
a &= 1 - \frac{(3\lambda+4)}{3} p_{xy}.
\end{align*}
Substituting into the expression for $a$, we get
\begin{align*}
e^{-\frac{2}{3} (3\lambda + 4)t_1}
&= 1 - \frac{(3\lambda+4)}{3} p_{xy}.
\end{align*}
Taking logarithms and solving for $t_1$ yields
\begin{align*}
t_1 &= \frac{\ln \left(1 - \frac{(3\lambda+4)}{3} p_{xy}\right)}{-\frac{2}{3}(3\lambda+4)}.
\end{align*}
Using $p_{xx}+p_{xy}=1$, this expression can equivalently be written as
\begin{align*}
t_1 &= \frac{\ln \left(p_{xx} - \frac{3\lambda+1}{3} p_{xy}\right)}{- \frac{2}{3} (3+4\lambda)}.
\end{align*}
Thus, for fixed $\lambda  > 0$, the value of $t_1$ is uniquely determined by the site pattern probabilities, and hence $t_1$ is generically identifiable.
\end{proof}

\begin{prop} \label{prop:3taxa:lambda-given-t}
Fix $t_1 > 0$. Then $\lambda$ is generically identifiable from the site pattern probability distribution.
\end{prop}

\begin{proof}
For fixed $t_1 > 0$, define
\begin{align*}
    g(\lambda) &= \frac{3-3 e^{-\frac{2}{3}  (3 \lambda +4) t_1}}{3 \lambda +4} = p_{xy}.
\end{align*}
Differentiating, we obtain
\begin{align*}
    g'(\lambda) &= \frac{6 t_1 e^{-\frac{2}{3}(3 \lambda +4) t_1}}{3 \lambda +4}-\frac{3 \left(3-3 e^{-\frac{2}{3} (3 \lambda +4) t_1}\right)}{(3 \lambda +4)^2}.
\end{align*}
Our aim is to show that $g'(\lambda) < 0$, implying that $g$ is strictly decreasing and therefore injective.
Let
\[ a = e^{-\frac{2}{3} ( 3 \lambda + 4)t_1 } \quad \text{and} \quad x = (3\lambda+4).\]
Then, 
\begin{align*}
    g'(\lambda) &= \frac{6 t_1 a}{x} - \frac{3(3-3a)}{x^2}
    = \frac{6 t_1 ax - 9(1-a)}{x^2}.
\end{align*}
So, the sign of $g'(\lambda)$ is the sign of the numerator:
\[ N \coloneqq 6 t_1 a x - 9(1-a).\]
Recall that $a = e^{-\frac{2}{3} ( 3 \lambda + 4)t_1} = e^{-\frac{2}{3} t_1 x} $ and substitute this into $N$:
\begin{align*}
    N &= 6 t_1 x e^{-\frac{2}{3} ( 3 \lambda + 4)t_1 } - 9 (1-e^{-\frac{2}{3} ( 3 \lambda + 4)t_1 })
    = 3 \left[ 2 t_1 x e^{-\frac{2}{3} t_1 x} - 3 (1-e^{-\frac{2}{3} t_1 x})\right]
\end{align*}
Now, let $y = \frac{2}{3} t_1 x$. Then, $a = e^{-y} $ and $2t_1x = 3y$.
Hence,
\begin{align*}
    N &= 3 \left[ 3y e^{-y} - 3(1-e^{-y})\right]
    = 9 \left[y e^{-y} - (1-e^{-y}) \right]
    = 9 \left[ (y+1) e^{-y}-1\right.]
\end{align*}
It remains to show that $(y+1) e^{-y} < 1$, or equivalently, $e^y > y+1$. This inequality is true for all real numbers $y \neq 0$. Since $t_1 > 0$ and $\lambda > 0$ (which implies $x > 0$), we have $y =  \frac{2}{3} t_1 x > 0$. 
Therefore, $N < 0$, and hence, $g'(\lambda)<0$. Thus, $g$ is strictly decreasing and therefore injective. Hence, given $t_1>0$, the parameter $\lambda$ is generically identifiable.
\end{proof}

\section{The 3-taxon case}\label{sec:three}
We now consider a 3-taxon species tree with $t$-values $t_1 < t_2$ (Fig.~\ref{fig:transfer}(i)). There are three possible (rooted) gene tree topologies:
\[
T_1 = ((1,2),3), \quad T_2 = ((1,3),2), \quad T_3 = ((2,3),1),
\]
where $T_1$ matches the species tree topology.

\paragraph{Gene tree probabilities.}
As in the 2-taxon case, different gene tree histories are induced by a finite number of relevant transfer events. Each gene history determines a gene tree topology and associated event times.
A full classification of the gene histories and their densities in the 3-taxon case is provided in Appendices~\ref{app:classification} and~\ref{app:gene-densities}. By integrating over all possible gene histories and their associated divergence time densities, we obtain the following result.

\begin{prop} \label{prop:gene-tree-probs-3taxa-main}
    Let $\cS$ be the $3$-taxon species tree from Fig.~\ref{fig:transfer}(i), and let $\boldsymbol{t} = (t_1,t_2)$ with $t_1 < t_2$ denote the $t$-values of its two interior vertices.
    Under the standard LGT model with transfer rate $\lambda$, the probabilities of the three gene tree topologies $T_1 = ((1,2),3)$, $T_2 = ((1,3),2)$, and $T_3 = ((2,3),1)$ are given by
    \begin{align*}
         \mathbb{P}(T_1 | \mathcal{S},\boldsymbol{t}, \lambda) &= \frac{1}{3} + \frac{2}{3} e^{-3 \lambda t_1},\\
     \mathbb{P}(T_2 | \mathcal{S},\boldsymbol{t}, \lambda) &=   \mathbb{P}(T_3 | \mathcal{S},\boldsymbol{t}, \lambda) = \frac{1}{3} - \frac{1}{3} e^{-3 \lambda t_1}.
    \end{align*}
    In particular, the gene tree topology $T_1$ matching the species tree $\cS$ has the highest probability.
\end{prop}

We note that this result is already established in the literature~\cite[Proposition~9.4]{steel2016}. As the proof in~\cite{steel2016} does not rely on gene histories and their associated divergence time densities, we include an alternative proof in Appendix~\ref{app:gene-probs}.

\paragraph{Site pattern probabilities under the JC69 model.}
Analogously to the 2-taxon case, we obtain site pattern probabilities under the JC69 model by averaging over the distribution of gene histories and their associated divergence time densities.

\begin{prop}\label{prop:site-pattern-probabilities-3taxa}
Let $\cS$ be the $3$-taxon species tree from Fig.~\ref{fig:transfer}(i), and let $\boldsymbol{t} = (t_1,t_2)$ with $t_1 < t_2$ denote the $t$-values of its two interior vertices. 
Under the standard LGT model with transfer rate $\lambda$, the marginal site pattern probabilities under the JC69 model with $\mu = 4/3$ are given in closed form in Appendix~\ref{app:site-pattern-probabilities}, where the explicit expressions in $t_1,t_2$, and $\lambda$ are derived.
\end{prop}

\subsection{Identifiability results}
We now consider parameter identifiability in the standard LGT model for the 3-taxon case. We distinguish between identifiability based on gene tree frequencies and identifiability derived from site pattern probabilities under the JC69 model. 

\subsubsection{Identifiability of parameters from gene tree frequencies}
From Proposition \ref{prop:gene-tree-probs-3taxa-main}, the gene tree probability corresponding to the species tree is given by 
\[\frac{1}{3}+\frac{2}{3}e^{-3\lambda t_1},\] 
whereas each of the two alternative topologies has probability 
\[\frac{1}{3}-\frac{1}{3}e^{-3\lambda t_1}.\] 
These expressions can be used to establish identifiability of the species tree parameter when the gene tree probabilities are known. This result has already been noted in the literature (see, e.g.,~\cite{steel2013,steel2016,yu2023some}), although it has not always been explicitly formulated as an identifiability statement. In fact, the corresponding result holds in the more general setting of the so-called \emph{extended LGT model}~\cite{steel2013} and for species trees of arbitrary size~\cite[Corollary 3.2.1]{yu2023some}. More broadly, identifiability of the species tree topology under LGT models has been studied extensively, with additional results available (see, e.g.,~\cite{Daskalakis2015,roch2013,yu2023some}). In contrast, our focus below is on the identifiability of branch lengths and transfer rates from gene tree probabilities, as well as on parameter identifiability from site pattern probabilities. Nevertheless, we include the species tree topology identifiability result here for completeness.

\begin{theorem}
In the 3-taxon case, the species tree topology is generically identifiable from the distribution of gene tree topologies under the standard LGT model.
\end{theorem}

\begin{proof}
Let $T^\ast$ denote the topology of the true species tree, and let $T_1, T_2, T_3$ be the three possible gene tree topologies for three taxa. Let $\mathbb{P}(T_i)$ denote the probability of observing gene tree topology $T_i$ under the LGT model.

It follows from Proposition \ref{prop:gene-tree-probs-3taxa-main} that the true topology has strictly higher probability than each of the alternative topologies. In particular,
\[\mathbb{P}(T^\ast) > \mathbb{P}(T_i) \quad \text{ for } \quad T_i \neq T^\ast,\]
so that \[T^\ast = \argmax_{T_i \in \{T_1, T_2, T_3\}} \mathbb{P}(T_i).\]

Hence, the true species tree topology is uniquely identified as the gene tree topology with maximum probability.
\end{proof}

We now consider identifiability of the parameters along a fixed species tree, namely $t_1, t_2,$ and $\lambda.$ Note that the gene tree probabilities derived in Proposition \ref{prop:gene-tree-probs-3taxa-main} do not depend on $t_2$; 
hence the 
$t$-values of the root and $t_2$ are not generically identifiable from the gene tree distribution.

The gene tree probabilities do depend on $\lambda$ and $t_1$, but only through their product, as seen in the expressions in Proposition~\ref{prop:gene-tree-probs-3taxa-main}. Moreover, since the three gene tree probabilities must sum to one, they are subject to an additional constraint. Taken together, these observations lead to the following result.

\begin{theorem}
Let $\cS$ be the 3-taxon species tree from Fig.~\ref{fig:transfer}(i)). Under the standard LGT model with transfer rate $\lambda$, the only generically identifiable parameter from the gene tree distribution is the composite parameter $\lambda t_1$. In particular, $t_1$ and $\lambda$ are not separately generically identifiable, and $t_2$ is not generically identifiable.
\end{theorem}

\begin{proof}
Let the probability of the true species tree topology be denoted $p_1$, where 
\[ p_1 = \frac{1}{3} + \frac{2}{3} e^{-3 \lambda t_1},\] 
and let the probability of each of the two alternative tree topologies on 3 taxa be denoted by 
\[ p_2 = \frac{1}{3} - \frac{1}{3} e^{-3 \lambda t_1}.\]  
Note that $p_1 + 2p_2 =1$.  

Solving for $\lambda t_1$ from either expression yields
\[ \lambda t_1 = -\frac{1}{3}\ln (1-3p_2) = -\frac{1}{3} \ln \biggl(\frac{3}{2}p_1-\frac{1}{2}\biggr).\]
Thus, only the composite parameter $\lambda t_1$ is generically identifiable from the gene tree distribution; the parameters $\lambda$ and $t_1$ cannot be separately identified.
\end{proof}

We now consider the case in which a multilocus data set is available and a gene tree has been estimated for each locus. In species tree inference under the multispecies coalescent model, it is common to treat such gene tree estimates as observed data for downstream inference of the species tree. Adopting a similar approach under the LGT model, let $f_1, f_2,$ and $f_3$ denote the observed frequencies of the three gene tree topologies.

The species tree topology can then be estimated by selecting the topology with the largest observed frequency. If the LGT rate $\lambda$ is known, the parameter $t_1$ can be estimated as
\[\hat{t}_1 = -\frac{1}{3\lambda} \ln \left( \frac{3}{2} f_1 - \frac{1}{2} \right).\]
Alternatively, if $t_1$ is known, the LGT rate $\lambda$ can be estimated from the same relationship.

\subsubsection{Identifiability of parameters from site pattern probabilities}

The mutation process along gene trees, which gives rise to sequence data observed at the leaves, provides additional information that can be used to estimate the species tree topology and model parameters. In this section, we investigate which parameters of the standard LGT model are identifiable from site pattern probabilities under the JC69 model in the 3-taxon case.

We begin by showing that the species tree topology is identifiable.

\begin{theorem}
Let $\cS$ be the 3-taxon species tree from Fig.~\ref{fig:transfer}(i)). Under the standard LGT model with transfer rate $\lambda$, the site pattern probabilities under the JC69 model satisfy
\[p_{xxy} > p_{xyx} = p_{yxx}.\]
Consequently, the species tree topology is generically identifiable from the site pattern probability distribution by selecting the topology corresponding to
\[\argmax\{ p_{xxy},\, p_{xyx},\, p_{yxx} \}.\]
\end{theorem}

The proof follows an analogous argument to the proof of Theorem 3a in~\cite{zhu2021}.
\begin{proof} 
   Each marginal site pattern probability is obtained by averaging over the 14 gene histories classified in Appendix~\ref{app:classification}. For the argument below, note that gene histories with subscript $1$ favor the site pattern $xxy$ in the sense that the induced probability of $xxy$ is strictly larger than that of $xyx$ and $yxx$, while the latter two are equal. Similarly, gene histories with subscript $2$ favor $xyx$, and those with subscript $3$ favor $yxx$.
    
   Each marginal site pattern probability is obtained by averaging over the 14 gene histories classified in Appendix~\ref{app:classification}. By Proposition~\ref{prop:gene-tree-densities}, the six gene histories $G_{1a}$, $G_{1b}$, $G_{2a}$, $G_{2b}$, $G_{3a}$, and $G_{3b}$ all have the same densities. Similarly, $G_{1c}$, $G_{2c}$, and $G_{3c}$ share a common density. Taken together, the contributions of these nine gene histories to the site pattern $xxy$ are equal to their contributions to the site patterns $xyx$ and $yxx$.

   In contrast, for gene histories of type $G_{1x}$ or $G_{1y}$, the induced probability of $xxy$ is strictly greater than that of either $xyx$ or $yxx$, while $p_{xyx} = p_{yxx}$ still holds.

    It remains to consider the three gene histories $G_{1d}, G_{2d},$ and $G_{3d}$. By Proposition~\ref{prop:gene-tree-densities}, $G_{2d}$ and $G_{3d}$ have the same density, while the density of $G_{1d}$ is twice as large. Hence, the induced site pattern probability for $xxy$ from $G_{1d}$ is strictly larger than the corresponding probabilities for $xyx$ and $yxx$ from $G_{2d}$ and $G_{3d}$, respectively, while the latter two remain equal.

    Combining all cases yields
    \[ p_{xxy} > p_{xyx} = p_{yxx},\]
    as required.
\end{proof}

We now turn to the local identifiability of the parameters $t_1, t_2$, and $\lambda$ of the standard LGT model for species tree $\cS$ in Fig.~\ref{fig:transfer}(i) from site pattern probabilities under the JC69 model.

\begin{definition}[Local identifiability]
Let $\theta = (t_1,t_2,\lambda)$ denote the parameter vector of the standard LGT model with transfer rate $\lambda$ for species tree $\cS$ in Fig.~\ref{fig:transfer}(i), and let
\[ f(\theta) = (p_{xxx}, p_{xxy}, p_{xyy}, p_{xyz})\]
denote the corresponding vector of site pattern probabilities under the JC69 model.

The parameter $\theta$ is said to be \emph{locally identifiable} at $\theta_0$ if there exists a neighborhood $U$ of $\theta_0$ such that, for all $\theta \in U$,
\[f(\theta) = f(\theta_0) \quad \Longrightarrow \quad \theta = \theta_0.\]
Equivalently, $f$ is {\it locally injective} at $\theta_0$.
\end{definition}

Intuitively, local identifiability means that distinct parameter values sufficiently close to the true value produce distinct distributions.

\begin{definition}[Generic local identifiability]
The parameter vector $\theta$ is said to be \emph{generically locally identifiable} if it is locally identifiable at all points of a nonempty open dense subset of the parameter space.
\end{definition}

\smallskip
A standard sufficient condition for local identifiability is that the Jacobian matrix $J_f$ has full column rank (see, e.g., \citet[Proposition 16.1.7]{sullivant2023}). By the inverse function theorem, this implies that the parameter-to-distribution map is locally injective.

To establish generic local identifiability, it suffices to show that $J_f$ has full column rank on a nonempty open dense subset of the parameter space (see, e.g.,~\citet{Gross2022}). Equivalently, it suffices to show that at least one maximal minor of $J_f$ is not identically zero.

\begin{prop}
Let 
\[
f(t_1,t_2,\lambda) =  (p_{xxx}, p_{xxy}, p_{xyy}, p_{xyz})\]
denote the vector of site pattern probabilities under the JC69 model induced by the standard LGT model with transfer rate $\lambda > 0$ for species tree $\cS$ in Fig.~\ref{fig:transfer}(i) with $0 < t_1 < t_2$.

Then the parameter vector $(t_1,t_2,\lambda)$ is generically locally identifiable; equivalently, the Jacobian matrix $J_f$ has generic rank $3$.
\end{prop}

\begin{proof}
We compute the Jacobian matrix $J_f$ of $f$ with respect to $(t_1,t_2,\lambda)$, which is a $4 \times 3$ matrix.

A symbolic computation in \texttt{Mathematica}~\cite{mathematica} shows that
\[
\mathrm{rank}\left(J_f\big|_{(t_1,t_2,\lambda)=(1,2,0.5)}\right)=3.
\]
Since $(1,2,0.5)$ satisfies $0<t_1<t_2$ and $\lambda>0$, the Jacobian attains full column rank at an admissible parameter point. Therefore at least one $3\times 3$ minor of $J_f$ is nonzero at this point, and hence is not identically zero as a function of $(t_1,t_2,\lambda)$.

It follows that $J_f$ has generic rank $3$, i.e., rank $3$ on a nonempty open dense subset of
\[
\Theta = \{(t_1,t_2,\lambda): 0<t_1<t_2,\ \lambda>0\}.
\]
Consequently, the parameter vector $(t_1,t_2,\lambda)$ is generically locally identifiable.

The \texttt{Mathematica} script used to perform this computation is provided alongside this manuscript.
\end{proof}

\medskip
While we have established that the parameters of the LGT model are generically locally identifiable, we conjecture that a stronger result holds.

\begin{conjecture}\label{conj:identifiability}
Let $\mathcal{S}$ be the 3-taxon species tree shown in Fig.~\ref{fig:transfer}(i). Under the standard LGT model with transfer rate $\lambda$, the parameters $t_1$, $t_2$, and $\lambda$ are generically identifiable from site pattern probabilities under the JC69 model.
\end{conjecture}

\section{Parameter estimation in the 3-taxon case}\label{sec:estimate}
An important goal in establishing parameter identifiability is to determine which model parameters can be estimated from empirical data. Having established local identifiability of the parameters $(t_1, t_2, \lambda)$ in the 3-taxon case and conjectured generic identifiability, we next consider estimation of these parameters from observed sequence data using maximum likelihood.

 Without loss of generality, let the tree in Fig. \ref{fig:transfer}(i) be the species tree (which is assumed to be fixed and known), and let $p_{ijk}$ be the probability that species $a$ has nucleotide $i$, species $b$ has nucleotide $j$, and species $c$ has nucleotide $k$ at a particular site in the sequence alignment under the standard LGT model. Let $\mathbf{p} = (p_{xxx}, p_{xxy}, p_{xyx}, p_{yxx}, p_{xyz})$ be the vector of site pattern probabilities, and let the observed frequencies in the data be denoted by $\mathbf{f} = (f_{xxx}, f_{xxy}, f_{xyx}, f_{yxx}, f_{xyz})$, where $f_{ijk}$ if the number of sites for which species $a$ has nucleotide $i$, species $b$ has nucleotide $j$, and species $c$ has nucleotide $k$.  The data $\mathbf{f}$ have a multinomial distribution with parameters $\mathbf{p}$ and $M$, where $M$ is the number of sites in the observed sequence alignment. Note that $\mathbf{p}$ is a function of the parameter vector $(t_1, t_2, \lambda)$. Thus the log of the likelihood function is given by
$$\log L(t_1, t_2, \lambda| \mathbf{f}) \propto f_{xxx} \log(p_{xxx}) + f_{xxy} \log(p_{xxy}) + (f_{xyx}+f_{yxx}) \log(p_{xyx}) + f_{xyz} \log(p_{xyz}). $$
This log likelihood can be maximized with respect to the parameter vector $(t_1, t_2, \lambda)$
to obtain the maximum likelihood estimators (MLEs) of the parameters, which are known to have favorable statistical properties, such as consistency and asymptotic normality.

Note that the parameters are subject to the constraints that $\lambda >0, t_1 > 0, t_2 > 0,$ and $ t_2 > t_1.$ Because unconstrained optimization is often numerically more efficient than constrained optimization, we follow the approach taken in Peng et al. (2022)\nocite{swofford2022} described in \cite{swofford2022} in transforming the parameter space so that unconstrained optimization algorithms can be applied. Specifically, we define
$$ z_1 = \sin^{-1}\biggl(\sqrt{\frac{t_1}{t_2}}\biggr)$$
$$ z_2 = \log (t_2)$$
$$ z_3 = \log(\lambda)$$
and note that $z_i \in \mathbb{R}$ for $i = 1, 2, 3.$ 

We implemented our MLE calculations in the R software \cite{R}, which allows us to take advantage of R's built-in optimization functions. Specifically, we used R's \texttt{optim} function with the default method due to Nelder and Mead (1965)\nocite{neldermead1965} to obtain parameter estimates.  The Nelder-Mead method uses only function values and does not require derivative information. We note that derivatives would be straightforward to calculate and our  optimization routines could likely be made more efficient by incorporating derivative information. Since the calculations without derivatives  take only seconds to compute, we have used that approach here.  Following optimization, the parameter values are transformed back to the original units by 
$$t_1 = t_2 \sin^2(z_1)$$
$$t_2 = e^{\log(z_2)}$$
$$\lambda = e^{\log(z_3)}.$$

To evaluate the performance of the MLEs, we carried out two simulation studies. In the first simulation study, we simulated data under the model by generating a sample of 1,000,000 sites from a 5-category multinomial distribution with category probabilities given by $\mathbf{p} = (p_{xxx},p_{xxy},p_{xyx}, p_{yxx}, p_{xyz})$. This corresponds to the generation of an alignment of length 1,000,000 bp from the species tree in Fig. \ref{fig:transfer}(i) under the LGT model in which each site has its own gene tree topology. We then obtained the MLEs from the simulated data, and repeated the entire procedure 100 times.  We considered  $\lambda = 0.8 $ and $\lambda= 1.2$, and for each of these, we considered three branch length combinations: $(t_1, t_2) = (0.1, 0.2); (0.2, 0.4);$ and $(0.4, 0.8).$

Histograms of the MLEs for each of the parameters when $\lambda=1.2$ are shown in Fig. \ref{fig:sim1-results} (the corresponding results when $\lambda=0.8$ are show in Appendix~\ref{appendix:simulation}).  Two observations are immediately clear from the simulation study. First, the parameter estimates appear to be unbiased since the histograms for each of the parameters are centered at the true values. This is highlighted by the close agreement between the true parameter value (black vertical line in each plot) and the average of the 100 MLEs (gray vertical line in each plot). Second, the distribution of the parameter estimates appears to be approximately normal, as expected.  In the case where $t_1$ and $t_2$ are short (subfigures (a), (b), and (c)), the histograms appear ``spiky'' because one replicate produced outlying observations that expanded the range of the $x$-axis.

\begin{figure}[htbp]
    \centering
    \includegraphics[scale=0.6]{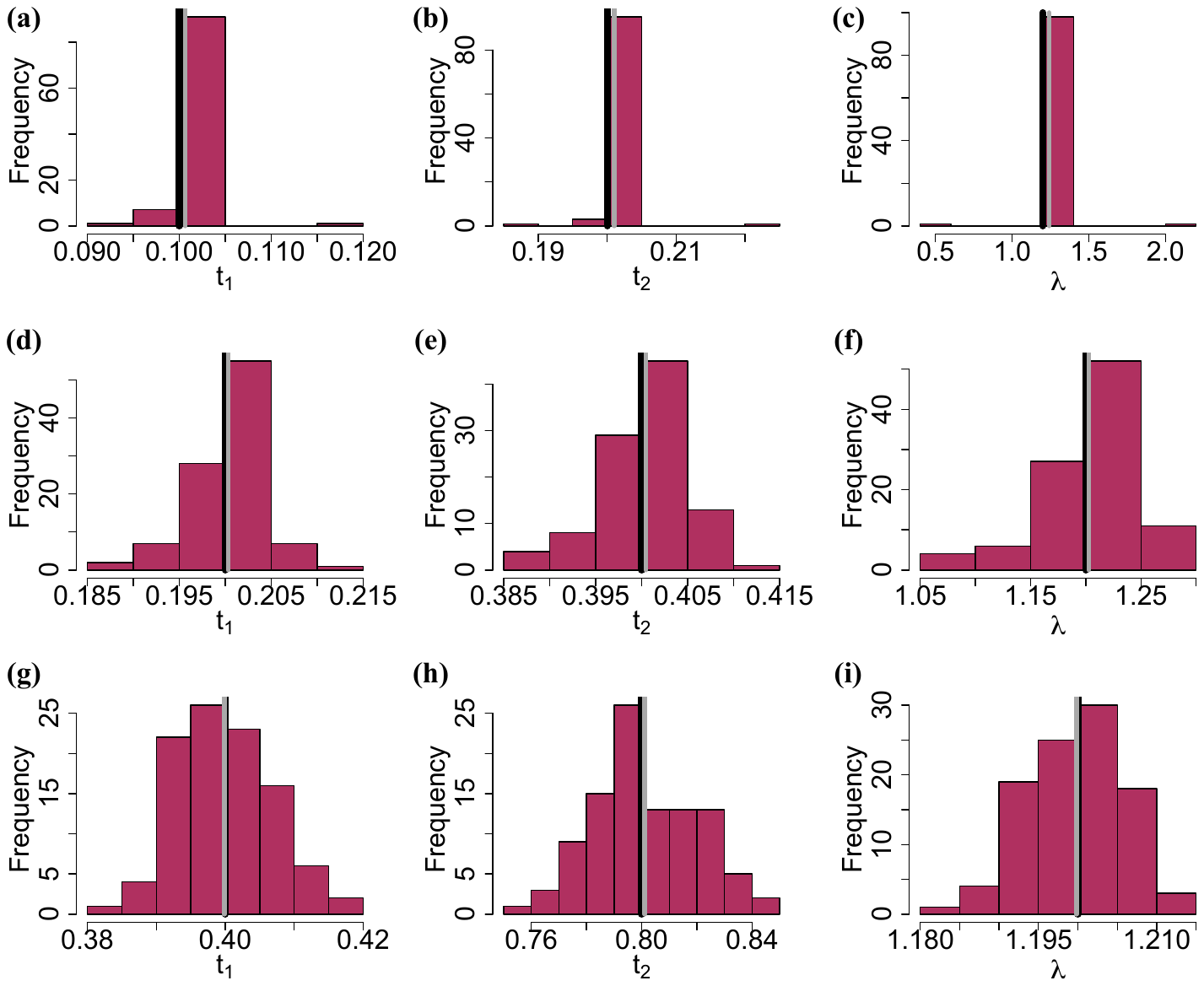}
    \caption{Results of the first simulation study when $\lambda=1.2$.  (a), (b), (c) The first row shows histograms of the 100 MLEs for parameters $t_1$, $t_2$, and $\lambda$, respectively, when the true values are $t_1=0.1$ and $t_2=0.2$. (d), (e), (f)  The second row shows histograms of the 100 MLEs for parameters $t_1$, $t_2$, and $\lambda$, respectively, when the true values are $t_1=0.2$ and $t_2=0.4$. (g), (h), (i) The third row shows histograms of the 100 MLEs for parameters $t_1$, $t_2$, and $\lambda$, respectively, when the true values are $t_1=0.4$ and $t_2=0.8$. In all plots, the black vertical line denotes the true parameter value and the gray vertical line gives the average estimate over the 100 replicates. }
    \label{fig:sim1-results}
\end{figure}

In the second simulation study, we considered the effect of the sample size (i.e., the length of the alignment) as a way of assessing consistency of the MLEs.  The simulation was carried out in the same way as the first simulation study, except that the sample size $M$ was varied so that the following values were considered: $M = 1000,5000,10000,50000,100000,500000,$ $1000000,5000000,10000000$. For each value of $M$ and each combination of branch lengths and $\lambda$, we simulated 1000 replicate data sets and computed MLEs for each parameter. The root mean squared error (RMSE) was then computed as 
\begin{equation}
RMSE(\theta) = \sqrt{\frac{1}{1000}\sum_{i=1}^{1000} (\hat{\theta} - \theta)^2}
\end{equation}
where $\theta$ is one of the parameters (e.g., either $t_1$ , $t_2$, or $\lambda$) and $\hat{\theta}$ is the MLE of $\theta$. We then plotted the RMSE as a function of the logarithm of the sample size $M$. Our expectation is that the RMSE decreases as $M$ increases, which, together with the observation that the MLEs are unbiased, indicates statistical consistency. 

The results of the second simulation study are shown in Fig. \ref{fig:rmse}. It is clear that the RMSE decreases as the sample size increases. In general, the error in estimating $\lambda$ is larger than that in estimating $t_1$ and $t_2$, though there are exceptions. The RMSE for all parameters is very low when the sample size is large.

\begin{figure}[htbp]
    \centering
    \includegraphics[scale=0.6]{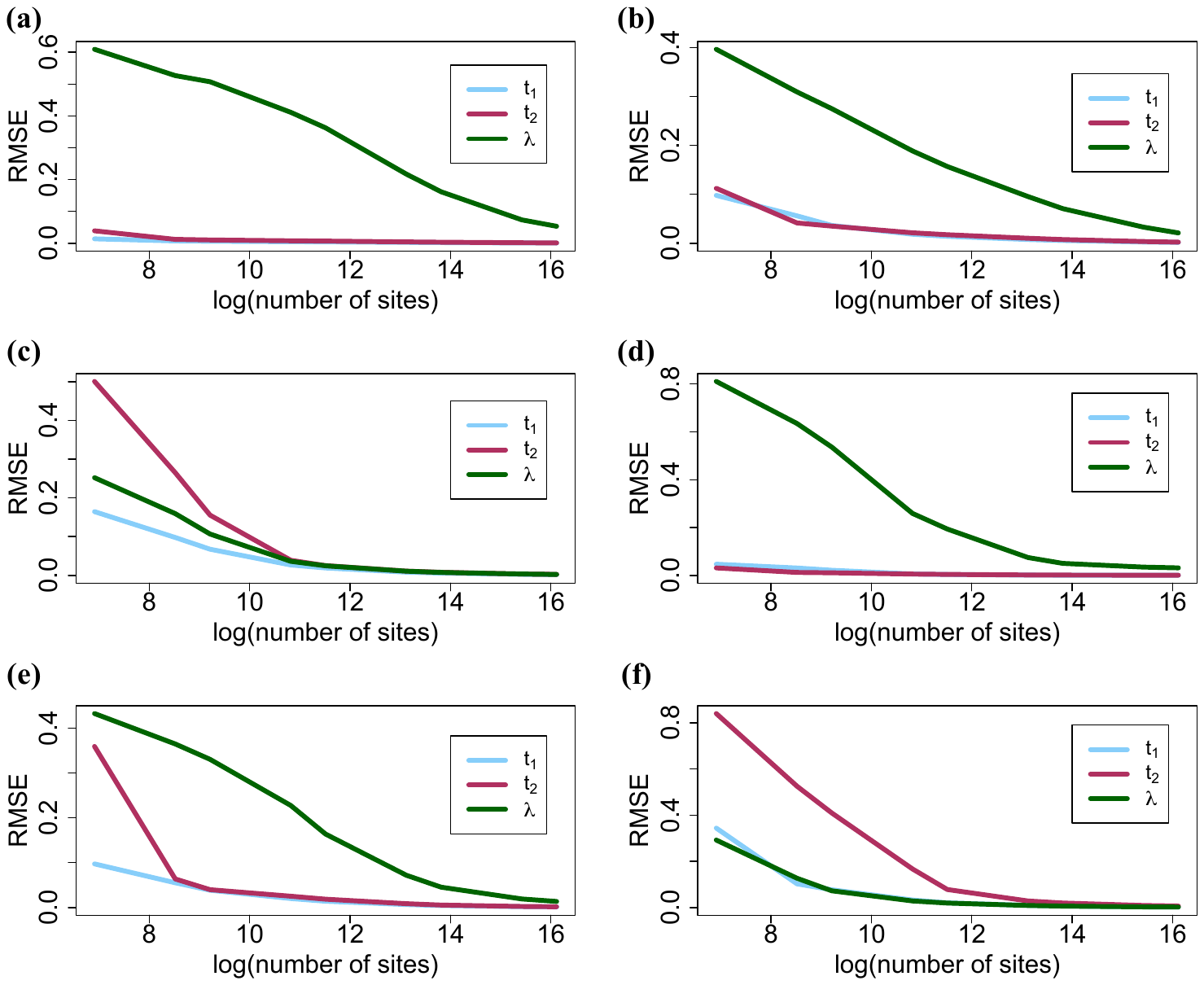}
    \caption{Results of the second simulation study.  Each plot shows the RMSE for estimates of each parameter as a function of the logarithm of the number of sites when the true values of the parameters $(t_1, t_2, \lambda)$ are (a) (0.1, 0.2, 0.8); (b) (0.1, 0.2, 1.2); (c) (0.2, 0.4, 0.8); (d) (0.2, 0.4, 1.2); (e) (0.4, 0.8, 0.8); (f) (0.4, 0.8, 1.2).  }
    \label{fig:rmse}
\end{figure}

\section{Distinguishing LGT and ILS in the 3-taxon case}\label{sec:ILS}

\paragraph{Background on the multispecies coalescent (MSC) model.} Incomplete lineage sorting (ILS) is another evolutionary process that can cause the gene tree and species tree topologies to be incongruent.  The multispecies coalescent (MSC)~\cite{kubatko2023species,mirarab2021multispecies} is typically used to model ILS. Consider a $3$-taxon species tree $\cS$. As in the LGT case, $t_2$ denotes the $t$-value of the root of $\cS$ and $t_1$ denotes the $t$-value of the interior vertex of $\cS$ that is not the root such that  the unique non-pendant arc of $\cS$ has length $t_2-t_1>0$. In the MSC model, a gene tree $T$ on the three lineages $1$, $2$, and $3$ can be thought of as if it evolves within $\cS$ following a backwards-in-time process in which lineages of $T$ coalesce as they are followed back towards the root of $\cS$. Note that lineages sometimes fail to coalesce in the way prescribed by $\cS$. For example, lineages $1$ and $2$ can fail to coalesce along the non-pendant arc of $\cS$, and lineages $1$ and $3$ can instead coalesce further back in time in which case the resulting gene tree has topology $((1,3),2)$. Let $T_1=((1,2),3)$, $T_2=((1,3),2)$, and $T_3=((2,3),1)$ be the three gene tree topologies on three taxa.  
It can be shown that the (species tree) matching gene tree topology $T_1$ and each of the two non-matching gene tree topologies $T_2$ and $T_3$ have the following probabilities, where $\boldsymbol{t} = (t_1,t_2)$ with $t_1<t_2$:

\begin{align*}
\mathbb{P}_{\text{MSC}}(T_1\mid\cS,\boldsymbol{t} )&=1-\frac{2}{3} e^{-(t_2-t_1)}\\
\mathbb{P}_{\text{MSC}}(T_2\mid\cS,\boldsymbol{t} )=\mathbb{P}_{\text{MSC}}(T_3\mid\cS,\boldsymbol{t})&=\frac{1}{3} e^{-(t_2-t_1)}.
\end{align*}
 We refer the reader to~\cite[Chapter 1]{kubatko2023species} for an introduction to the MSC in the 3-taxon case, and to~\cite{degnan2005gene} for the gene tree probabilities for an $n$-taxon species tree with $n\geq 3$.

\paragraph{Comparison of the gene tree probabilities under the MSC and LGT model.} Let $\cS$ be the $3$-taxon species tree as shown in Fig. \ref{fig:transfer}(i), and let $t_1$ and $t_2$ with $t_1<t_2$ be the $t$-values of the two interior vertices of $\cS$. It is well known that, unless $t_2-t_1=0$, the probability of the matching gene tree topology is always larger than the probability of either of the two non-matching gene tree topologies under the MSC  model. Intuitively, the shorter the non-pendant arc of $\cS$, the more likely it is for a coalescent event to happen between lineages $1$ and $3$, or between $2$ and $3$. Turning to the LGT model, it follows from the probabilities of the three gene tree topologies given in Proposition~\ref{prop:gene-tree-probs-3taxa-main} that the probability of the matching gene tree topology is always larger than the probability of either of the two non-matching gene tree topologies. In particular, the larger $\lambda  t_1$, the more likely it is for the first relevant transfer to result in a non-matching gene tree topology. A comparison of the  probabilities of the matching and non-matching gene tree topologies under both models is shown in Fig.~\ref{fig:comparing-probs}.

\begin{figure}[h!]
\noindent\begin{minipage}[t]{0.5\textwidth}
\scalebox{0.5}{
\begin{tikzpicture}[x=0.95pt,y=0.95pt]
\definecolor{fillColor}{RGB}{255,255,255}
\path[use as bounding box,fill=fillColor,fill opacity=0.00] (0,0) rectangle (433.62,289.08);
\begin{scope}
\path[clip] (  0.00,  0.00) rectangle (433.62,289.08);
\definecolor{drawColor}{RGB}{255,255,255}
\definecolor{fillColor}{RGB}{255,255,255}

\path[draw=drawColor,line width= 0.6pt,line join=round,line cap=round,fill=fillColor] (  0.00,  0.00) rectangle (433.62,289.08);
\end{scope}
\begin{scope}
\path[clip] ( 52.21, 39.69) rectangle (428.12,263.07);
\definecolor{fillColor}{gray}{0.92}

\path[fill=fillColor] ( 52.21, 39.69) rectangle (428.12,263.07);
\definecolor{drawColor}{RGB}{255,255,255}

\path[draw=drawColor,line width= 0.3pt,line join=round] ( 52.21, 75.23) --
	(428.12, 75.23);

\path[draw=drawColor,line width= 0.3pt,line join=round] ( 52.21,126.00) --
	(428.12,126.00);

\path[draw=drawColor,line width= 0.3pt,line join=round] ( 52.21,176.76) --
	(428.12,176.76);

\path[draw=drawColor,line width= 0.3pt,line join=round] ( 52.21,227.53) --
	(428.12,227.53);

\path[draw=drawColor,line width= 0.3pt,line join=round] (126.25, 39.69) --
	(126.25,263.07);

\path[draw=drawColor,line width= 0.3pt,line join=round] (240.17, 39.69) --
	(240.17,263.07);

\path[draw=drawColor,line width= 0.3pt,line join=round] (354.08, 39.69) --
	(354.08,263.07);

\path[draw=drawColor,line width= 0.6pt,line join=round] ( 52.21, 49.85) --
	(428.12, 49.85);

\path[draw=drawColor,line width= 0.6pt,line join=round] ( 52.21,100.61) --
	(428.12,100.61);

\path[draw=drawColor,line width= 0.6pt,line join=round] ( 52.21,151.38) --
	(428.12,151.38);

\path[draw=drawColor,line width= 0.6pt,line join=round] ( 52.21,202.15) --
	(428.12,202.15);

\path[draw=drawColor,line width= 0.6pt,line join=round] ( 52.21,252.91) --
	(428.12,252.91);

\path[draw=drawColor,line width= 0.6pt,line join=round] ( 69.30, 39.69) --
	( 69.30,263.07);

\path[draw=drawColor,line width= 0.6pt,line join=round] (183.21, 39.69) --
	(183.21,263.07);

\path[draw=drawColor,line width= 0.6pt,line join=round] (297.12, 39.69) --
	(297.12,263.07);

\path[draw=drawColor,line width= 0.6pt,line join=round] (411.03, 39.69) --
	(411.03,263.07);
\definecolor{drawColor}{RGB}{176,48,96}

\path[draw=drawColor,line width= 2.3pt,line join=round] ( 69.30, 49.85) --
	( 72.72, 53.79) --
	( 76.13, 57.50) --
	( 79.55, 61.00) --
	( 82.97, 64.29) --
	( 86.39, 67.39) --
	( 89.80, 70.31) --
	( 93.22, 73.06) --
	( 96.64, 75.65) --
	(100.06, 78.09) --
	(103.47, 80.39) --
	(106.89, 82.55) --
	(110.31, 84.59) --
	(113.72, 86.51) --
	(117.14, 88.31) --
	(120.56, 90.02) --
	(123.98, 91.62) --
	(127.39, 93.13) --
	(130.81, 94.55) --
	(134.23, 95.89) --
	(137.65, 97.15) --
	(141.06, 98.34) --
	(144.48, 99.45) --
	(147.90,100.51) --
	(151.32,101.50) --
	(154.73,102.43) --
	(158.15,103.31) --
	(161.57,104.14) --
	(164.98,104.92) --
	(168.40,105.65) --
	(171.82,106.35) --
	(175.24,107.00) --
	(178.65,107.61) --
	(182.07,108.19) --
	(185.49,108.73) --
	(188.91,109.25) --
	(192.32,109.73) --
	(195.74,110.18) --
	(199.16,110.61) --
	(202.58,111.02) --
	(205.99,111.40) --
	(209.41,111.75) --
	(212.83,112.09) --
	(216.24,112.41) --
	(219.66,112.71) --
	(223.08,112.99) --
	(226.50,113.25) --
	(229.91,113.50) --
	(233.33,113.74) --
	(236.75,113.96) --
	(240.17,114.17) --
	(243.58,114.36) --
	(247.00,114.55) --
	(250.42,114.72) --
	(253.84,114.88) --
	(257.25,115.04) --
	(260.67,115.18) --
	(264.09,115.32) --
	(267.51,115.45) --
	(270.92,115.57) --
	(274.34,115.69) --
	(277.76,115.79) --
	(281.17,115.90) --
	(284.59,115.99) --
	(288.01,116.08) --
	(291.43,116.17) --
	(294.84,116.25) --
	(298.26,116.32) --
	(301.68,116.39) --
	(305.10,116.46) --
	(308.51,116.52) --
	(311.93,116.58) --
	(315.35,116.64) --
	(318.77,116.69) --
	(322.18,116.74) --
	(325.60,116.78) --
	(329.02,116.83) --
	(332.43,116.87) --
	(335.85,116.91) --
	(339.27,116.94) --
	(342.69,116.98) --
	(346.10,117.01) --
	(349.52,117.04) --
	(352.94,117.07) --
	(356.36,117.10) --
	(359.77,117.12) --
	(363.19,117.15) --
	(366.61,117.17) --
	(370.03,117.19) --
	(373.44,117.21) --
	(376.86,117.23) --
	(380.28,117.25) --
	(383.69,117.26) --
	(387.11,117.28) --
	(390.53,117.30) --
	(393.95,117.31) --
	(397.36,117.32) --
	(400.78,117.33) --
	(404.20,117.35) --
	(407.62,117.36) --
	(411.03,117.37);
\definecolor{drawColor}{RGB}{0,100,0}

\path[draw=drawColor,line width= 2.3pt,line join=round] ( 69.30,252.91) --
	( 72.72,245.03) --
	( 76.13,237.61) --
	( 79.55,230.61) --
	( 82.97,224.03) --
	( 86.39,217.83) --
	( 89.80,211.99) --
	( 93.22,206.49) --
	( 96.64,201.31) --
	(100.06,196.43) --
	(103.47,191.83) --
	(106.89,187.51) --
	(110.31,183.43) --
	(113.72,179.59) --
	(117.14,175.98) --
	(120.56,172.58) --
	(123.98,169.37) --
	(127.39,166.35) --
	(130.81,163.51) --
	(134.23,160.83) --
	(137.65,158.31) --
	(141.06,155.94) --
	(144.48,153.70) --
	(147.90,151.59) --
	(151.32,149.61) --
	(154.73,147.74) --
	(158.15,145.98) --
	(161.57,144.33) --
	(164.98,142.77) --
	(168.40,141.30) --
	(171.82,139.91) --
	(175.24,138.61) --
	(178.65,137.38) --
	(182.07,136.23) --
	(185.49,135.14) --
	(188.91,134.11) --
	(192.32,133.15) --
	(195.74,132.24) --
	(199.16,131.38) --
	(202.58,130.58) --
	(205.99,129.82) --
	(209.41,129.10) --
	(212.83,128.43) --
	(216.24,127.79) --
	(219.66,127.20) --
	(223.08,126.63) --
	(226.50,126.10) --
	(229.91,125.61) --
	(233.33,125.14) --
	(236.75,124.69) --
	(240.17,124.28) --
	(243.58,123.88) --
	(247.00,123.51) --
	(250.42,123.17) --
	(253.84,122.84) --
	(257.25,122.53) --
	(260.67,122.24) --
	(264.09,121.96) --
	(267.51,121.71) --
	(270.92,121.46) --
	(274.34,121.23) --
	(277.76,121.02) --
	(281.17,120.82) --
	(284.59,120.63) --
	(288.01,120.45) --
	(291.43,120.28) --
	(294.84,120.12) --
	(298.26,119.97) --
	(301.68,119.82) --
	(305.10,119.69) --
	(308.51,119.57) --
	(311.93,119.45) --
	(315.35,119.34) --
	(318.77,119.23) --
	(322.18,119.13) --
	(325.60,119.04) --
	(329.02,118.95) --
	(332.43,118.87) --
	(335.85,118.79) --
	(339.27,118.72) --
	(342.69,118.65) --
	(346.10,118.59) --
	(349.52,118.52) --
	(352.94,118.47) --
	(356.36,118.41) --
	(359.77,118.36) --
	(363.19,118.31) --
	(366.61,118.27) --
	(370.03,118.23) --
	(373.44,118.19) --
	(376.86,118.15) --
	(380.28,118.11) --
	(383.69,118.08) --
	(387.11,118.05) --
	(390.53,118.02) --
	(393.95,117.99) --
	(397.36,117.96) --
	(400.78,117.94) --
	(404.20,117.91) --
	(407.62,117.89) --
	(411.03,117.87);
\end{scope}
\begin{scope}
\path[clip] (  0.00,  0.00) rectangle (433.62,289.08);
\definecolor{drawColor}{gray}{0.30}

\node[text=drawColor,anchor=base east,inner sep=0pt, outer sep=0pt, scale=  1.40] at ( 47.26, 45.03) {0.00};

\node[text=drawColor,anchor=base east,inner sep=0pt, outer sep=0pt, scale=  1.40] at ( 47.26, 95.79) {0.25};

\node[text=drawColor,anchor=base east,inner sep=0pt, outer sep=0pt, scale=  1.40] at ( 47.26,146.56) {0.50};

\node[text=drawColor,anchor=base east,inner sep=0pt, outer sep=0pt, scale=  1.40] at ( 47.26,197.33) {0.75};

\node[text=drawColor,anchor=base east,inner sep=0pt, outer sep=0pt, scale=  1.40] at ( 47.26,248.09) {1.00};
\end{scope}
\begin{scope}
\path[clip] (  0.00,  0.00) rectangle (433.62,289.08);
\definecolor{drawColor}{gray}{0.20}

\path[draw=drawColor,line width= 0.6pt,line join=round] ( 49.46, 49.85) --
	( 52.21, 49.85);

\path[draw=drawColor,line width= 0.6pt,line join=round] ( 49.46,100.61) --
	( 52.21,100.61);

\path[draw=drawColor,line width= 0.6pt,line join=round] ( 49.46,151.38) --
	( 52.21,151.38);

\path[draw=drawColor,line width= 0.6pt,line join=round] ( 49.46,202.15) --
	( 52.21,202.15);

\path[draw=drawColor,line width= 0.6pt,line join=round] ( 49.46,252.91) --
	( 52.21,252.91);
\end{scope}
\begin{scope}
\path[clip] (  0.00,  0.00) rectangle (433.62,289.08);
\definecolor{drawColor}{gray}{0.20}

\path[draw=drawColor,line width= 0.6pt,line join=round] ( 69.30, 36.94) --
	( 69.30, 39.69);

\path[draw=drawColor,line width= 0.6pt,line join=round] (183.21, 36.94) --
	(183.21, 39.69);

\path[draw=drawColor,line width= 0.6pt,line join=round] (297.12, 36.94) --
	(297.12, 39.69);

\path[draw=drawColor,line width= 0.6pt,line join=round] (411.03, 36.94) --
	(411.03, 39.69);
\end{scope}
\begin{scope}
\path[clip] (  0.00,  0.00) rectangle (433.62,289.08);
\definecolor{drawColor}{gray}{0.30}

\node[text=drawColor,anchor=base,inner sep=0pt, outer sep=0pt, scale=  1.40] at ( 69.30, 25.10) {0};

\node[text=drawColor,anchor=base,inner sep=0pt, outer sep=0pt, scale=  1.40] at (183.21, 25.10) {2};

\node[text=drawColor,anchor=base,inner sep=0pt, outer sep=0pt, scale=  1.40] at (297.12, 25.10) {4};

\node[text=drawColor,anchor=base,inner sep=0pt, outer sep=0pt, scale=  1.40] at (411.03, 25.10) {6};
\end{scope}
\begin{scope}
\path[clip] (  0.00,  0.00) rectangle (433.62,289.08);
\definecolor{drawColor}{RGB}{0,0,0}

\node[text=drawColor,anchor=base,inner sep=0pt, outer sep=0pt, scale=  1.60] at (240.17,  8.61) {\itshape $\lambda  t_1$};
\end{scope}
\begin{scope}
\path[clip] (  0.00,  0.00) rectangle (433.62,289.08);
\definecolor{drawColor}{RGB}{0,0,0}

\node[text=drawColor,rotate= 90.00,anchor=base,inner sep=0pt, outer sep=0pt, scale=  1.60] at ( 16.52,151.38) {probability};
\end{scope}
\begin{scope}
\path[clip] (  0.00,  0.00) rectangle (433.62,289.08);
\definecolor{drawColor}{RGB}{0,0,0}

\node[text=drawColor,anchor=base,inner sep=0pt, outer sep=0pt, scale=  1.70] at (240.17,271.87) {(a) LGT};
\end{scope}
\end{tikzpicture}}
\end{minipage}
\begin{minipage}[t]{0.5\textwidth}
\scalebox{0.5}{
\begin{tikzpicture}[x=0.95pt,y=0.95pt]
\definecolor{fillColor}{RGB}{255,255,255}
\path[use as bounding box,fill=fillColor,fill opacity=0.00] (0,0) rectangle (433.62,289.08);
\begin{scope}
\path[clip] (  0.00,  0.00) rectangle (433.62,289.08);
\definecolor{drawColor}{RGB}{255,255,255}
\definecolor{fillColor}{RGB}{255,255,255}

\path[draw=drawColor,line width= 0.6pt,line join=round,line cap=round,fill=fillColor] (  0.00,  0.00) rectangle (433.62,289.08);
\end{scope}
\begin{scope}
\path[clip] ( 52.21, 39.69) rectangle (428.12,263.07);
\definecolor{fillColor}{gray}{0.92}

\path[fill=fillColor] ( 52.21, 39.69) rectangle (428.12,263.07);
\definecolor{drawColor}{RGB}{255,255,255}

\path[draw=drawColor,line width= 0.3pt,line join=round] ( 52.21, 75.13) --
	(428.12, 75.13);

\path[draw=drawColor,line width= 0.3pt,line join=round] ( 52.21,126.02) --
	(428.12,126.02);

\path[draw=drawColor,line width= 0.3pt,line join=round] ( 52.21,176.91) --
	(428.12,176.91);

\path[draw=drawColor,line width= 0.3pt,line join=round] ( 52.21,227.80) --
	(428.12,227.80);

\path[draw=drawColor,line width= 0.3pt,line join=round] (126.25, 39.69) --
	(126.25,263.07);

\path[draw=drawColor,line width= 0.3pt,line join=round] (240.17, 39.69) --
	(240.17,263.07);

\path[draw=drawColor,line width= 0.3pt,line join=round] (354.08, 39.69) --
	(354.08,263.07);

\path[draw=drawColor,line width= 0.6pt,line join=round] ( 52.21, 49.68) --
	(428.12, 49.68);

\path[draw=drawColor,line width= 0.6pt,line join=round] ( 52.21,100.57) --
	(428.12,100.57);

\path[draw=drawColor,line width= 0.6pt,line join=round] ( 52.21,151.46) --
	(428.12,151.46);

\path[draw=drawColor,line width= 0.6pt,line join=round] ( 52.21,202.36) --
	(428.12,202.36);

\path[draw=drawColor,line width= 0.6pt,line join=round] ( 52.21,253.25) --
	(428.12,253.25);

\path[draw=drawColor,line width= 0.6pt,line join=round] ( 69.30, 39.69) --
	( 69.30,263.07);

\path[draw=drawColor,line width= 0.6pt,line join=round] (183.21, 39.69) --
	(183.21,263.07);

\path[draw=drawColor,line width= 0.6pt,line join=round] (297.12, 39.69) --
	(297.12,263.07);

\path[draw=drawColor,line width= 0.6pt,line join=round] (411.03, 39.69) --
	(411.03,263.07);
\definecolor{drawColor}{RGB}{176,48,96}

\path[draw=drawColor,line width= 2.3pt,line join=round] ( 69.30,117.54) --
	( 72.72,113.58) --
	( 76.13,109.86) --
	( 79.55,106.36) --
	( 82.97,103.06) --
	( 86.39, 99.95) --
	( 89.80, 97.02) --
	( 93.22, 94.26) --
	( 96.64, 91.67) --
	(100.06, 89.22) --
	(103.47, 86.92) --
	(106.89, 84.75) --
	(110.31, 82.71) --
	(113.72, 80.78) --
	(117.14, 78.97) --
	(120.56, 77.27) --
	(123.98, 75.66) --
	(127.39, 74.15) --
	(130.81, 72.72) --
	(134.23, 71.38) --
	(137.65, 70.12) --
	(141.06, 68.93) --
	(144.48, 67.81) --
	(147.90, 66.75) --
	(151.32, 65.76) --
	(154.73, 64.82) --
	(158.15, 63.94) --
	(161.57, 63.11) --
	(164.98, 62.33) --
	(168.40, 61.59) --
	(171.82, 60.90) --
	(175.24, 60.24) --
	(178.65, 59.63) --
	(182.07, 59.05) --
	(185.49, 58.50) --
	(188.91, 57.99) --
	(192.32, 57.50) --
	(195.74, 57.05) --
	(199.16, 56.62) --
	(202.58, 56.22) --
	(205.99, 55.83) --
	(209.41, 55.48) --
	(212.83, 55.14) --
	(216.24, 54.82) --
	(219.66, 54.52) --
	(223.08, 54.24) --
	(226.50, 53.97) --
	(229.91, 53.72) --
	(233.33, 53.49) --
	(236.75, 53.27) --
	(240.17, 53.06) --
	(243.58, 52.86) --
	(247.00, 52.68) --
	(250.42, 52.50) --
	(253.84, 52.34) --
	(257.25, 52.18) --
	(260.67, 52.04) --
	(264.09, 51.90) --
	(267.51, 51.77) --
	(270.92, 51.65) --
	(274.34, 51.53) --
	(277.76, 51.42) --
	(281.17, 51.32) --
	(284.59, 51.23) --
	(288.01, 51.14) --
	(291.43, 51.05) --
	(294.84, 50.97) --
	(298.26, 50.90) --
	(301.68, 50.83) --
	(305.10, 50.76) --
	(308.51, 50.70) --
	(311.93, 50.64) --
	(315.35, 50.58) --
	(318.77, 50.53) --
	(322.18, 50.48) --
	(325.60, 50.43) --
	(329.02, 50.39) --
	(332.43, 50.35) --
	(335.85, 50.31) --
	(339.27, 50.27) --
	(342.69, 50.24) --
	(346.10, 50.20) --
	(349.52, 50.17) --
	(352.94, 50.15) --
	(356.36, 50.12) --
	(359.77, 50.09) --
	(363.19, 50.07) --
	(366.61, 50.05) --
	(370.03, 50.02) --
	(373.44, 50.00) --
	(376.86, 49.99) --
	(380.28, 49.97) --
	(383.69, 49.95) --
	(387.11, 49.93) --
	(390.53, 49.92) --
	(393.95, 49.91) --
	(397.36, 49.89) --
	(400.78, 49.88) --
	(404.20, 49.87) --
	(407.62, 49.86) --
	(411.03, 49.85);
\definecolor{drawColor}{RGB}{0,100,0}

\path[draw=drawColor,line width= 2.3pt,line join=round] ( 69.30,117.54) --
	( 72.72,125.44) --
	( 76.13,132.88) --
	( 79.55,139.89) --
	( 82.97,146.49) --
	( 86.39,152.71) --
	( 89.80,158.57) --
	( 93.22,164.08) --
	( 96.64,169.27) --
	(100.06,174.16) --
	(103.47,178.77) --
	(106.89,183.11) --
	(110.31,187.19) --
	(113.72,191.04) --
	(117.14,194.66) --
	(120.56,198.07) --
	(123.98,201.29) --
	(127.39,204.31) --
	(130.81,207.16) --
	(134.23,209.85) --
	(137.65,212.37) --
	(141.06,214.75) --
	(144.48,217.00) --
	(147.90,219.11) --
	(151.32,221.10) --
	(154.73,222.97) --
	(158.15,224.73) --
	(161.57,226.39) --
	(164.98,227.96) --
	(168.40,229.43) --
	(171.82,230.82) --
	(175.24,232.12) --
	(178.65,233.35) --
	(182.07,234.51) --
	(185.49,235.60) --
	(188.91,236.63) --
	(192.32,237.60) --
	(195.74,238.51) --
	(199.16,239.37) --
	(202.58,240.18) --
	(205.99,240.94) --
	(209.41,241.66) --
	(212.83,242.33) --
	(216.24,242.97) --
	(219.66,243.57) --
	(223.08,244.13) --
	(226.50,244.66) --
	(229.91,245.16) --
	(233.33,245.63) --
	(236.75,246.08) --
	(240.17,246.49) --
	(243.58,246.89) --
	(247.00,247.26) --
	(250.42,247.61) --
	(253.84,247.94) --
	(257.25,248.24) --
	(260.67,248.54) --
	(264.09,248.81) --
	(267.51,249.07) --
	(270.92,249.31) --
	(274.34,249.54) --
	(277.76,249.76) --
	(281.17,249.96) --
	(284.59,250.15) --
	(288.01,250.33) --
	(291.43,250.50) --
	(294.84,250.66) --
	(298.26,250.81) --
	(301.68,250.96) --
	(305.10,251.09) --
	(308.51,251.21) --
	(311.93,251.33) --
	(315.35,251.45) --
	(318.77,251.55) --
	(322.18,251.65) --
	(325.60,251.74) --
	(329.02,251.83) --
	(332.43,251.91) --
	(335.85,251.99) --
	(339.27,252.06) --
	(342.69,252.13) --
	(346.10,252.20) --
	(349.52,252.26) --
	(352.94,252.32) --
	(356.36,252.37) --
	(359.77,252.42) --
	(363.19,252.47) --
	(366.61,252.52) --
	(370.03,252.56) --
	(373.44,252.60) --
	(376.86,252.64) --
	(380.28,252.67) --
	(383.69,252.71) --
	(387.11,252.74) --
	(390.53,252.77) --
	(393.95,252.80) --
	(397.36,252.82) --
	(400.78,252.85) --
	(404.20,252.87) --
	(407.62,252.89) --
	(411.03,252.91);
\end{scope}
\begin{scope}
\path[clip] (  0.00,  0.00) rectangle (433.62,289.08);
\definecolor{drawColor}{gray}{0.30}

\node[text=drawColor,anchor=base east,inner sep=0pt, outer sep=0pt, scale=  1.40] at ( 47.26, 44.86) {0.00};

\node[text=drawColor,anchor=base east,inner sep=0pt, outer sep=0pt, scale=  1.40] at ( 47.26, 95.75) {0.25};

\node[text=drawColor,anchor=base east,inner sep=0pt, outer sep=0pt, scale=  1.40] at ( 47.26,146.64) {0.50};

\node[text=drawColor,anchor=base east,inner sep=0pt, outer sep=0pt, scale=  1.40] at ( 47.26,197.54) {0.75};

\node[text=drawColor,anchor=base east,inner sep=0pt, outer sep=0pt, scale=  1.40] at ( 47.26,248.43) {1.00};
\end{scope}
\begin{scope}
\path[clip] (  0.00,  0.00) rectangle (433.62,289.08);
\definecolor{drawColor}{gray}{0.20}

\path[draw=drawColor,line width= 0.6pt,line join=round] ( 49.46, 49.68) --
	( 52.21, 49.68);

\path[draw=drawColor,line width= 0.6pt,line join=round] ( 49.46,100.57) --
	( 52.21,100.57);

\path[draw=drawColor,line width= 0.6pt,line join=round] ( 49.46,151.46) --
	( 52.21,151.46);

\path[draw=drawColor,line width= 0.6pt,line join=round] ( 49.46,202.36) --
	( 52.21,202.36);

\path[draw=drawColor,line width= 0.6pt,line join=round] ( 49.46,253.25) --
	( 52.21,253.25);
\end{scope}
\begin{scope}
\path[clip] (  0.00,  0.00) rectangle (433.62,289.08);
\definecolor{drawColor}{gray}{0.20}

\path[draw=drawColor,line width= 0.6pt,line join=round] ( 69.30, 36.94) --
	( 69.30, 39.69);

\path[draw=drawColor,line width= 0.6pt,line join=round] (183.21, 36.94) --
	(183.21, 39.69);

\path[draw=drawColor,line width= 0.6pt,line join=round] (297.12, 36.94) --
	(297.12, 39.69);

\path[draw=drawColor,line width= 0.6pt,line join=round] (411.03, 36.94) --
	(411.03, 39.69);
\end{scope}
\begin{scope}
\path[clip] (  0.00,  0.00) rectangle (433.62,289.08);
\definecolor{drawColor}{gray}{0.30}

\node[text=drawColor,anchor=base,inner sep=0pt, outer sep=0pt, scale=  1.40] at ( 69.30, 25.10) {0};

\node[text=drawColor,anchor=base,inner sep=0pt, outer sep=0pt, scale=  1.40] at (183.21, 25.10) {2};

\node[text=drawColor,anchor=base,inner sep=0pt, outer sep=0pt, scale=  1.40] at (297.12, 25.10) {4};

\node[text=drawColor,anchor=base,inner sep=0pt, outer sep=0pt, scale=  1.40] at (411.03, 25.10) {6};
\end{scope}
\begin{scope}
\path[clip] (  0.00,  0.00) rectangle (433.62,289.08);
\definecolor{drawColor}{RGB}{0,0,0}

\node[text=drawColor,anchor=base,inner sep=0pt, outer sep=0pt, scale=  1.60] at (240.17,  8.61) {\itshape $t_2-t_1$};
\end{scope}
\begin{scope}
\path[clip] (  0.00,  0.00) rectangle (433.62,289.08);
\definecolor{drawColor}{RGB}{0,0,0}

\node[text=drawColor,rotate= 90.00,anchor=base,inner sep=0pt, outer sep=0pt, scale=  1.60] at ( 16.52,151.38) {probability};
\end{scope}
\begin{scope}
\path[clip] (  0.00,  0.00) rectangle (433.62,289.08);
\definecolor{drawColor}{RGB}{0,0,0}

\node[text=drawColor,anchor=base,inner sep=0pt, outer sep=0pt, scale=  1.70] at (240.17,271.87) {(b) MSC};
\end{scope}
\end{tikzpicture}}
\end{minipage}
\caption{Probabilities of the matching (green) and non-matching (maroon) gene tree topologies under the LGT model (a) and the MSC model (b) plotted as a function of $\lambda t_1$ and $t_2-t_1$, respectively.}
\label{fig:comparing-probs}
\end{figure}
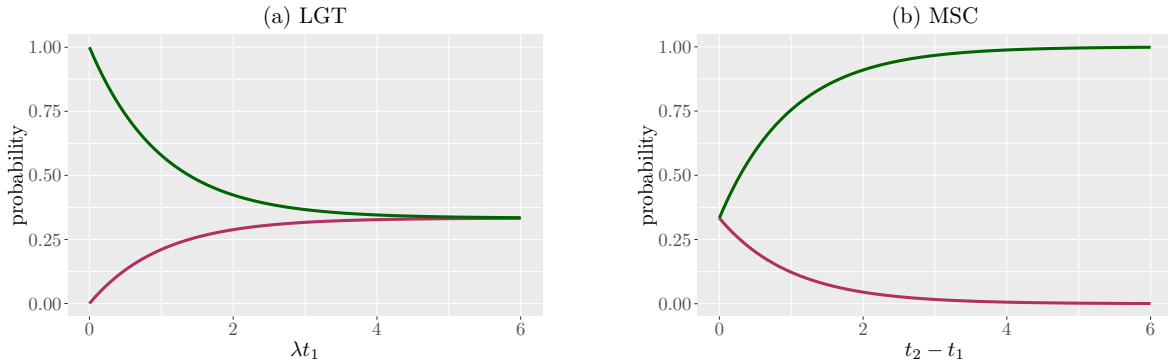

\paragraph{Indistinguishability between the LGT and MSC models.}
Given a a distribution of gene trees, we will next see that it is not always possible to determine if the distribution has been obtained under the standard LGT model or under the MSC model. The indistinguishability between these two models arises for certain parameter combinations. For a 3-taxon species tree, the parameters of interest are its two divergence times and, for the LGT model, the transfer rate. Let  $\cS$ and  $\cS'$ be a $3$-taxon species tree with the same topology, let $\boldsymbol{t} = (t_1,t_2)$ with $0<t_1 < t_2$ be the $t$-values of the two interior vertices of $\cS$, and let $\boldsymbol{t'} = (t_1',t_2')$ with $0<t_1' < t_2'$ be the $t$-values of the two interior vertices of $\cS'$. Furthermore, let $T_1$ be the gene tree topology that matches $\cS$ (and $\cS'$). Since 
$$\mathbb{P_{\text{LGT}}}(T_1 | \mathcal{S},\boldsymbol{t}, \lambda) = \frac{1}{3} + \frac{2}{3} e^{-3 \lambda t_1} \text{ and }\mathbb{P}_{\text{MSC}}(T_1\mid\cS',\boldsymbol{t'} )=1-\frac{2}{3} e^{-(t_2'-t_1')}$$
have range $(\frac 1 3,1)$, it follows that for every value for $\boldsymbol{t}$ with $\mathbb{P_{\text{LGT}}}(T_1 | \mathcal{S},\boldsymbol{t}, \lambda)$, there exists a choice for $\boldsymbol{t'}$ such that $\mathbb{P_{\text{LGT}}}(T_1 | \mathcal{S},\boldsymbol{t}, \lambda)=\mathbb{P}_{\text{MSC}}(T_1\mid\cS',\boldsymbol{t'} )$.

We next establish a stronger indistinguishability result when $\boldsymbol{t} =\boldsymbol{t'}$. 

\begin{theorem}
Let $\cS$ be a species tree on three taxa, and let $\boldsymbol{t} = (t_1,t_2)$ with $0<t_1 < t_2$ be the $t$-values of its two interior vertices. 
Furthermore, let $\lambda>0$ be a fixed transfer rate. There are infinitely many pairs $t_1$ and $t_2$ such that the standard LGT model and the MSC model are indistinguishable.
\end{theorem}

\begin{proof}
Let $T$ be a gene tree topology that is different from the topology of $\cS$. Then, under the LGT model, the probability of $T$ is $\mathbb{P}_{\text{LGT}}(T \mid \mathcal{S}, \boldsymbol{t},\lambda)=\frac{1}{3} - \frac{1}{3} e^{-3 \lambda t_1}$, and under the MSC model the probability of $T$ is $\mathbb{P}_{\text{MSC}}(T \mid \mathcal{S},\boldsymbol{t})=\frac 1 3 e^{-(t_2-t_1)}$. For the LGT and MSC model to be indistinguishable, we have 
\begin{align}
\mathbb{P}_{\text{LGT}}(T \mid \mathcal{S},\boldsymbol{t}, \lambda)=\mathbb{P}_{\text{MSC}}(T \mid \mathcal{S},\boldsymbol{t})&\iff \frac{1}{3} - \frac{1}{3} e^{-3 \lambda t_1}=\frac 1 3 e^{-(t_2-t_1)}\nonumber\\
&\iff 1-e^{-3\lambda t_1}=e^{-(t_2-t_1)}\nonumber\\
&\iff t_2=t_1-\ln{(1-e^{-3\lambda t_1})}.\label{eq:one}
\end{align}
We complete the proof by showing that, for a fixed $t_1$ (and a fixed $\lambda$), Eq.~\eqref{eq:one} satisfies the inequality $t_2>t_1$. First assume that $t_1=t_2$. Then, by Eq.~\eqref{eq:one}, we have  $$t_1=t_1-\ln{(1-e^{-3\lambda t_1})}\iff 0=\ln{(1-e^{-3\lambda t_1})}.$$ Since $e^{-3\lambda t_1}\ne 0$, it follows that $\ln{(1-e^{-3\lambda t_1})}\ne 0$; thereby contradicting our assumption. Second assume that $t_1>t_2$. Again, by  Eq.~\eqref{eq:one}, we have $t_1-t_2=\ln(1-e^{-3\lambda t_1})$. The left-hand side of the last equality is positive, whereas the right-hand side is negative because $e^{-3\lambda t_1}$ is in the interval $(0,1)$; another contradiction. By combining both cases, we have $t_2>t_1$. In particular, for any $t_1$ with $t_1>0$, there exists a $t_2$ with $t_2>t_1$ such that the $\mathbb{P}_{\text{LGT}}(T \mid \mathcal{S},\boldsymbol{t}, \lambda)=\mathbb{P}_{\text{MSC}}(T \mid \mathcal{S},\boldsymbol{t})$.
\end{proof}

\section{Discussion}
Given the ubiquity with which lateral gene transfer occurs in natural populations \cite{sieber2017lateral,soucy2015horizontal}, modeling this process is crucial in obtaining accurate phylogenetic inference.  We take several important steps in this direction. First, for models with two or three taxa at a time, we prove parameter identifiability (and importantly, we consider cases in which not all parameters are simultaneously identifiable). In contrast to earlier work on the identifiability of the species tree topology from gene tree probabilities, we focus on the numerical parameters of the model (such as branch lengths and transfer rates) and study their identifiability both from gene tree and site pattern probabilities. In the case of estimating the continuous parameters along a fixed tree for three taxa, we prove generic local identifiability, but conjecture generic identifiability, while in all other cases for two and three taxa, we are able to prove generic identifiability. Second, we show that maximum likelihood estimation is feasible under the standard LGT model described here in the 3-taxon case, and that parameter estimates satisfy the expected theoretical properties of consistency and asymptotic normality. Finally, we highlight the indistinguishability of the LGT process from the commonly-modeled ILS process, thus providing an important caution for those seeking to use observed gene tree frequencies to infer biological mechanism.

One implication of the indistinguishability result is that methods that infer species trees by considering triplets are likely to perform well under LGT models, as was observed in \cite{davidson2015phylogenomic}, since the requirement that the gene tree with the same topology as the species tree has the highest probability is met for the LGT model as it is for the ILS model. An interesting follow-up question to the indistinguishability result is whether these models can be distinguished from the site pattern probability distribution. Though we have not addressed that question here, we conjecture that these models \emph{will} be distinguishable if site pattern probabilities are known. 

Our work can be extended in several important ways.  First, we could consider allowing the rate of LGT to vary across the tree, with variation across branch lengths likely to be an important aspect of the evolutionary process to model. We note that such models have been previously proposed \cite{steel2013}.  Second, we could incorporate models more complex than the JC69 model, including the possible addition of models with rate variation across sites. Finally, we could extend our results to four taxa. While that greatly increases the complexity of the computations because the number of gene tree histories is much greater, significant new information could be gained from quartets as opposed to triplets. Additionally, the quartet site pattern probabilities could be used to infer species tree for an arbitrarily large number of species by integrating them into an existing composite likelihood framework \cite{kubatkoetal2025}, which would be relatively straightforward.

More generally, models that incorporate multiple evolutionary processes simultaneously (e.g., \cite{lietal2021}) are important for capturing the full range of evolutionary processes at work. Examining questions of tree and parameter identifiability in such models will be crucial to the development of strategies for inference under such models. Our work addresses the question of identifiability in the simple setting of the LGT process on small trees, but in doing so, we highlight the complexities and potential approaches that can be applied in more complicated scenarios.

\section*{Acknowledgments}
This material is based upon work supported by the National Science Foundation under Grant No. DMS-1929284 while the authors were in residence at the Institute for Computational and Experimental Research in Mathematics in Providence, RI, during the `Theory, Methods, and Applications of Quantitative Phylogenomics' semester program. The authors also gratefully acknowledge the Banff International Research Station for Mathematical Innovation and Discovery (BIRS) for hosting the workshop `Novel Mathematical Paradigm for Phylogenomics (25w5333)', which facilitated this research. SL was also supported by the New Zealand Marsden Fund from Government funding, administered by the Royal Society Te Ap\=arangi New Zealand. Finally, the authors thank John A. Rhodes and Seth Sullivant for helpful discussions on concepts and proofs related to identifiability.

\section*{Author contributions}
All authors have made substantial intellectual contributions to this work. All authors have read and approved the final manuscript.

\section*{Conflicts of interest} 
The authors declare there are no conflicts of interest.

\section*{Data \& code availability}
The \texttt{Mathematica} code for reproducing all calculations and the R Markdown file for reproducing the simulation study are available on GitHub: \url{https://github.com/lkubatko/LGT-Model}.

\nocite{*}
\bibliographystyle{abbrvnat}
\bibliography{References.bib}

@inproceedings{Daskalakis2015,
  title = {Species Trees from Gene Trees Despite a High Rate of Lateral Genetic Transfer: A Tight Bound (Extended Abstract)},
  DOI = {10.1137/1.9781611974331.ch110},
  booktitle = {Proceedings of the Twenty-Seventh Annual ACM-SIAM Symposium on Discrete Algorithms},
  publisher = {Society for Industrial and Applied Mathematics},
  author = {Daskalakis,  Constantinos and Roch,  Sebastien},
  year = {2015},
  pages = {1621--1630}
}

@article{degnan2005gene,
  title={Gene tree distributions under the coalescent process},
  author={Degnan, James H and Salter, Laura A},
  journal={Evolution},
  volume={59},
  number={1},
  pages={24--37},
  year={2005}
}

@article{degnan2009gene,
  title = {Gene tree discordance,  phylogenetic inference and the multispecies coalescent},
  volume = {24},
  DOI = {10.1016/j.tree.2009.01.009},
  number = {6},
  journal = {Trends in Ecology \& Evolution},
  publisher = {Elsevier BV},
  author = {Degnan,  James H. and Rosenberg,  Noah A.},
  year = {2009},
  pages = {332–340}
}

@article{galtier2007,
  title={A model of horizontal gene transfer and the bacterial phylogeny problem},
  author={Galtier, Nicolas},
  journal={Systematic Biology},
  volume={56},
  number={4},
  pages={633--642},
  year={2007},
  publisher={Oxford University Press}
}

@misc{garcia2007,
  author       = {Luis David Garcia-Puente and Jacob Porter},
  title        = {Small Phylogenetic Trees},
  howpublished = {\url{https://www.coloradocollege.edu/aapps/ldg/small-trees/small-trees_0.html}},
  institution  = {Colorado College},
  year         = {2007},
  note         = {Contributors include Marta Casanellas, Serkan Hosten, Lior Pachter, Stacey Stokes, Bernd Sturmfels, and Seth Sullivant},
}

@article{jukescantor1969,
  title={Evolution of protein molecules},
  author={Jukes, Thomas H and Cantor, Charles R},
  journal={Mammalian Protein Metabolism},
  volume={3},
  pages={21--132},
  year={1969}
}

@book{kubatko2023species,
  title={Species Tree Inference: {A} Guide to Methods and Applications},
  author={Kubatko, Laura and Knowles, L Lacey},
  year={2023},
  publisher={Princeton University Press}
}

@article{kubatkoetal2025,
    author = {Kubatko, L. and Kong, S. and Webb, E. and Chen,Z.},
    year = 2025,
    title = {The promise of composite likelihood for species-level phylogenomic inference}, 
    journal = {Evolutionary Journal of the Linnean Society}, 
    volume = 4,
    number = 1,
    pages = {kzaf008}
}

@article{lietal2021, 
    title = {The multilocus multispecies coalescent: a flexible new model of gene family evolution},
    author = {Li, Q. and Scornavacca, C. and Galtier, N. and Chan, Y.-B.},
    journal = {Systematic Biology},
    year = 2021,
    volume=70,
    number = 4,
    pages={822-837}
}

@article{linz2007,
  title = {A Likelihood Framework to Measure Horizontal Gene Transfer},
  volume = {24},
  DOI = {10.1093/molbev/msm052},
  number = {6},
  journal = {Molecular Biology and Evolution},
  publisher = {Oxford University Press (OUP)},
  author = {Linz,  S. and Radtke,  A. and von Haeseler,  A.},
  year = {2007},
  pages = {1312–1319}
}

@article{sand2013,
  title = {The standard lateral gene transfer model is statistically consistent for pectinate four-taxon trees},
  volume = {335},
  DOI = {10.1016/j.jtbi.2013.07.002},
  journal = {Journal of Theoretical Biology},
  publisher = {Elsevier BV},
  author = {Sand,  Andreas and Sand,  Andreas and Steel,  Mike},
  year = {2013},
  pages = {295–298}
}

@article{steel2013,
  title = {Identifying a species tree subject to random lateral gene transfer},
  volume = {322},
  DOI = {10.1016/j.jtbi.2013.01.009},
  journal = {Journal of Theoretical Biology},
  publisher = {Elsevier BV},
  author = {Steel,  Mike and Linz,  Simone and Huson,  Daniel H. and Sanderson,  Michael J.},
  year = {2013},
  pages = {81–93}
}

@book{steel2016,
  title     = "Phylogeny: Discrete and Random Processes in Evolution",
  author    = "Steel, Mike",
  publisher = "Society for Industrial and Applied Mathematics",
  year      =  2016,
  address   = "Philadelphia, PA"
}

@article{yang1994,
  title = {Statistical Properties of the Maximum Likelihood Method of Phylogenetic Estimation and Comparison With Distance Matrix Methods},
  volume = {43},
  DOI = {10.1093/sysbio/43.3.329},
  number = {3},
  journal = {Systematic Biology},
  publisher = {Oxford University Press (OUP)},
  author = {Yang,  Z.},
  year = {1994},
  pages = {329–342}
}

@article{zhu2021,
  title = {Complexity of the simplest species tree problem},
  volume = {38},
  DOI = {10.1093/molbev/msab009},
  number = {9},
  journal = {Molecular Biology and Evolution},
  publisher = {Oxford University Press (OUP)},
  author = {Zhu,  Tianqi and Yang,  Ziheng},
  editor = {Su,  Bing},
  year = {2021},
  pages = {3993–4009}
}

@article{pengetal2021,
  title={A fast likelihood approach for estimation of large phylogenies from continuous trait data},
  author={Peng, Jing and Rajeevan, Haseena and Kubatko, Laura and RoyChoudhury, Arindam},
  journal={Molecular Phylogenetics and Evolution},
  volume={161},
  pages={107142},
  year={2021},
  publisher={Elsevier}
}

@article{swofford2022,
  title={Estimation of speciation times under the multispecies coalescent},
  author={Peng, Jing and Swofford, David L and Kubatko, Laura},
  journal={Bioinformatics},
  volume={38},
  number={23},
  pages={5182--5190},
  year={2022},
  publisher={Oxford University Press}
}

@article{neldermead1965,
  title={A simplex method for function minimization},
  author={Nelder, John A and Mead, Roger},
  journal={The Computer Journal},
  volume={7},
  number={4},
  pages={308--313},
  year={1965},
  publisher={The British Computer Society}
}

@article{mirarab2021multispecies,
  title={Multispecies coalescent: {T}heory and applications in phylogenetics},
  author={Mirarab, Siavash and Nakhleh, Luay and Warnow, Tandy},
  journal={Annual Review of Ecology, Evolution, and Systematics},
  volume={52},
  number={1},
  pages={247--268},
  year={2021},
  publisher={Annual Reviews}
}

@misc{mathematica,
  author = {{Wolfram Research, Inc.}},
  title = {Mathematica, {V}ersion 14.3},
  year = {2025},
  address = {Champaign, IL},
  url = {https://www.wolfram.com/mathematica}
}

@Manual{R,
    title = {R: A Language and Environment for Statistical Computing},
    author = {{R Core Team}},
    organization = {R Foundation for Statistical Computing},
    address = {Vienna, Austria},
    year = {2021},
    url = {https://www.R-project.org/},
}

@article{kubatko2026evolving,
  title={An evolving view of species tree inference},
  author={Kubatko, Laura},
  journal={Systematic Biology},
  volume={75},
  number={2},
  pages={195--207},
  year={2026},
  publisher={Oxford University Press}
}

@article{roch2013,
  title = {Recovering the Treelike Trend of Evolution Despite Extensive Lateral Genetic Transfer: A Probabilistic Analysis},
  volume = {20},
   DOI = {10.1089/cmb.2012.0234},
  number = {2},
  journal = {Journal of Computational Biology},
  publisher = {SAGE Publications},
  author = {Roch,  Sebastien and Snir,  Sagi},
  year = {2013},
  pages = {93–112}
}

@book{sullivant2023,
  author    = {Seth Sullivant},
  title     = {Algebraic Statistics},
  series     = {Graduate Studies in Mathematics},
  volume     = {194},
  publisher  = {American Mathematical Society},
  year       = {2023},
  isbn       = {978-1-4704-7510-9}
}

@article{soucy2015horizontal,
  title={Horizontal gene transfer: Building the web of life},
  author={Soucy, Shannon M and Huang, Jinling and Gogarten, Johann Peter},
  journal={Nature Reviews Genetics},
  volume={16},
  number={8},
  pages={472--482},
  year={2015},
  publisher={Nature Publishing Group UK London}
}

@article{davidson2015phylogenomic,
  title={Phylogenomic species tree estimation in the presence of incomplete lineage sorting and horizontal gene transfer},
  author={Davidson, Ruth and Vachaspati, Pranjal and Mirarab, Siavash and Warnow, Tandy},
  journal={BMC Genomics},
  volume={16},
  number={Suppl 10},
  pages={S1},
  year={2015},
  publisher={Springer}
}

@article{sieber2017lateral,
  title={Lateral gene transfer between prokaryotes and eukaryotes},
  author={Sieber, Karsten B and Bromley, Robin E and Hotopp, Julie C Dunning},
  journal={Experimental Cell Research},
  volume={358},
  number={2},
  pages={421--426},
  year={2017},
  publisher={Elsevier}
}

@book{yu2023some,
  title={Some inference problems on networks with applications},
  author={Yu, Shuqi},
  year={2023},
  publisher={PhD thesis, The University of Wisconsin-Madison}
}

@article{liu2022antimicrobial,
  title={Antimicrobial-induced horizontal transfer of antimicrobial resistance genes in bacteria: A mini-review},
  author={Liu, Gang and Thomsen, Line Elnif and Olsen, John Elmerdahl},
  journal={Journal of Antimicrobial Chemotherapy},
  volume={77},
  number={3},
  pages={556--567},
  year={2022},
  publisher={Oxford University Press}
}

@article{rannala2003bayes,
  title={Bayes estimation of species divergence times and ancestral population sizes using {DNA} sequences from multiple loci},
  author={Rannala, Bruce and Yang, Ziheng},
  journal={Genetics},
  volume={164},
  number={4},
  pages={1645--1656},
  year={2003},
  publisher={Oxford University Press}
}

@article{suchard2005stochastic,
  title={Stochastic models for horizontal gene transfer: taking a random walk through tree space},
  author={Suchard, Marc A},
  journal={Genetics},
  volume={170},
  number={1},
  pages={419--431},
  year={2005},
  publisher={Oxford University Press}
}

@article{yu2011coalescent,
  title={Coalescent histories on phylogenetic networks and detection of hybridization despite incomplete lineage sorting},
  author={Yu, Yun and Than, Cuong and Degnan, James H and Nakhleh, Luay},
  journal={Systematic Biology},
  volume={60},
  number={2},
  pages={138--149},
  year={2011},
  publisher={Oxford University Press}
}

@article{degnan2018modeling,
  title={Modeling hybridization under the network multispecies coalescent},
  author={Degnan, James H},
  journal={Systematic Biology},
  volume={67},
  number={5},
  pages={786--799},
  year={2018},
  publisher={Oxford University Press}
}

@article{Gross2022,
  title = {Identifiability of linear compartmental models: {T}he singular locus},
  volume = {133},
  DOI = {10.1016/j.aam.2021.102268},
  journal = {Advances in Applied Mathematics},
  publisher = {Elsevier BV},
  author = {Gross,  Elizabeth and Meshkat,  Nicolette and Shiu,  Anne},
  year = {2022},
  pages = {102268}
}

\appendix
\section{Further details on the 3-taxon case}

\subsection{Gene history classification for the 3-taxon case} \label{app:classification}
Let $\cS$ be the $3$-taxon species tree in Fig.~\ref{fig:transfer}(i), with arcs $e_1,e_2,e_3$ incident to leaves $1,2,3$ and interior times $t_1<t_2$. Let $\underline{\sigma}=(\sigma_1,\ldots,\sigma_k)$ be a transfer sequence, and let $\sigma_1^\ast,\sigma_2^\ast$ denote the first and second relevant transfers at times $h_1<h_2$ (when they exist).

\begin{obs} \label{obs:sigma1}
For any transfer sequence $\underline{\sigma}$ on $\cS$:
\begin{enumerate}[(a)]
    \item there are at most two relevant transfers; and 
    \item the first transfer $\sigma_1$ is always relevant (so $\sigma_1^\ast = \sigma_1$) and determines the topology of $T[\underline{\sigma}]$.
\end{enumerate}
\end{obs}

Part (a) follows from the fact that a rooted binary tree on three leaves has exactly two interior vertices, so at most two transfers can affect the induced gene history. For part (b), the most recent transfer $\sigma_1$ determines which pair of leaves forms a cherry in $T[\underline{\sigma}]$, thereby fixing its topology (see also \cite{steel2013}). \medskip

\medskip
We next summarize all possible gene histories. For ease of notation, if $(p_i, p_i')$ is the transfer arc associated with $\sigma_i$, where $p_i$ lies on arc $e_j$ and $p_i'$ lies on arc $e_k$, we write $e_j \to e_k$.

\medskip
Each gene history is determined by:
\begin{enumerate}[(i)]
    \item the number of relevant transfers,
    \item their timing relative to $t_1$ and $t_2$, and
    \item the corresponding transfer arcs $(e_i \to e_j)$.
\end{enumerate}

This yields the following cases, where $g_1 < g_2$ denote the $t$-values of the two interior vertices of $T[\underline{\sigma}]$:

\begin{enumerate}
    \item \textbf{No relevant transfer.}\\
    If no transfer occurs, then $T[\underline{\sigma}] = T_1$, denoted $G_{1x}$, with interior times $g_1=t_1$, $g_2=t_2$.

    \item \textbf{One relevant transfer in $(0,t_1)$ and no other relevant transfers.}\\
    Let $\sigma_1^\ast$ occur at time $h_1 \in (0,t_1)$. The possible outcomes group into three cases, each corresponding to two symmetric transfer directions:
    \begin{itemize}
    \item Transfer $e_1 \to e_2$: $T_1$, history $G_{1a}$, $g_1 = h_1$, $g_2 = t_2$;
    \item Transfer $e_2 \to e_1$: $T_1$, history $G_{1b}$, $g_1 = h_1$, $g_2 = t_2$;
    \item Transfer $e_1 \to e_3$: $T_2$, history $G_{2a}$, $g_1 = h_1$, $g_2 = t_1$;
    \item Transfer $e_3 \to e_1$: $T_2$, history $G_{2b}$, $g_1 = h_1$, $g_2 = t_2$;
    \item Transfer $e_2 \to e_3$: $T_3$, history $G_{3a}$, $g_1 = h_1$, $g_2 = t_1$.
    \item Transfer $e_3 \to e_2$: $T_3$, history $G_{3b}$, $g_1 = h_1$, $g_2 = t_2$;
    \end{itemize}
    
    \item \textbf{Two relevant transfers in $(0,t_1)$.} \\
    Let $\sigma_1^\ast$ occur at time $h_1 \in (0,t_1)$ and $\sigma_2^\ast$ at time $h_2 \in (h_1,t_1)$. Since $\sigma_1^\ast$ determines the topology of $T[\sigma]$, we distinguish cases only by the direction of the first transfer arc. For each such case, there are two possible choices for the second transfer arc; however, these yield identical gene histories and are therefore not treated separately.
       \begin{itemize}
           \item Transfers between $e_1$ and $e_2$ (i.e., $e_1 \to e_2$ or $e_2 \to e_1$): $T_1$, history $G_{1c}$, $g_1=h_1$, $g_2=h_2$;
           \item Transfers between $e_1$ and $e_3$ (i.e., $e_1 \to e_3$ or $e_3 \to e_1$): $T_2$, history $G_{2c}$, $g_1=h_1$, $g_2=h_2$;
           \item Transfers between $e_2$ and $e_3$ (i.e., $e_2 \to e_3$ or $e_3 \to e_2$): $T_3$, history $G_{3c}$, $g_1=h_1$, $g_2=h_2$;
       \end{itemize}

   \item \textbf{First relevant transfer in $(0,t_1)$ and second relevant transfer in $(t_1,t_2)$.} \\
    Let $\sigma_1^\ast$ occur at time $h_1 \in (0,t_1)$ and $\sigma_2^\ast$ at time $h_2 \in (t_1,t_2)$. Then the possible outcomes are:
    \begin{itemize}
    \item Transfers between $e_1$ and $e_2$ (i.e., $e_1 \to e_2$ or $e_2 \to e_1$): $T_1$, $G_{1d}$, $g_1 = h_1$, $g_2 = h_2$;
    \item Transfer from $e_3$ to $e_1$: $T_2$, $G_{2d}$, $g_1 = h_1$, $g_2 = h_2$;
    \item Transfer from $e_3$ to $e_2$: $T_3$, $G_{3d}$, $g_1 = h_1$, $g_2 = h_2$.
    \end{itemize}
    Transfers of the form $e_1 \to e_3$ or $e_2 \to e_3$ cannot occur in this setting, as they would preclude the existence of a second relevant transfer in $(t_1,t_2)$.

    \item \textbf{First relevant transfer in $(t_1,t_2)$.}\\
    If the first relevant transfer occurs after $t_1$, then $T[\underline{\sigma}] = T_1$, denoted $G_{1y}$, with $g_1 = t_1$ and $g_2=h_1$.
\end{enumerate}

\subsection{Gene density derivations} \label{app:gene-densities}
We now derive the densities of the transfer times for all of the gene histories. We exclude the case of no relevant transfers, corresponding to $G_{1x}$, from the statement, as it corresponds to a point mass rather than a density; its probability is given by $e^{-\lambda (t_1+2t_2)}$.

\begin{prop}\label{prop:gene-tree-densities}
    Let $\cS$ be the $3$-taxon species tree from Fig.~\ref{fig:transfer}(i), and let $\boldsymbol{t} = (t_1,t_2)$ with $t_1 < t_2$ denote the $t$-values of its two interior vertices. Further, let $\underline{\sigma} = (\sigma_1, \ldots, \sigma_k)$ be a transfer sequence on $\cS$ with relevant transfers $\sigma_1^\ast$ and $\sigma_2^\ast$.
    Let $\boldsymbol{h} = (h_1, h_2)$ be the $t$-values of $\sigma_1^\ast$ and $\sigma_2^\ast$. 
    Under the standard LGT model with transfer rate $\lambda$, the joint densities for a gene history and transfer times given $\mathcal{S}$, $\boldsymbol{t}$, and $\lambda$ fall into the following cases:

    \begin{enumerate}[(i)]

    \item \textbf{One relevant transfer in $(0,t_1)$ and no other relevant transfers:}
        \begin{align*}
        f_{G_{1a}}(\boldsymbol{h}|\cS,\boldsymbol{t},\lambda) &= \frac{1}{2} \lambda e^{- \lambda (2 t_2 + h_1)} \quad \text{with} \quad 0 < h_1 < t_1, 
        \end{align*}
        with identical expressions for $G_{1b}$, $G_{2a}$, $G_{2b}$, $G_{3a}$, and $G_{3b}$.
   
    \item \textbf{Two relevant transfers in $(0,t_1)$:}
        \begin{align*}
            f_{G_{1c}}(\boldsymbol{h}|\cS,\boldsymbol{t},\lambda) &= 2 \lambda^2 e^{-\lambda (3h_1 + 2h_2)} \quad \text{with} \quad 0 < h_1 < t_1 \text{ and } 0 < h_2 < t_1-h_1, 
        \end{align*}
        with identical expressions for $G_{2c}$ and $G_{3c}$.

    \item \textbf{First relevant transfer in $(0,t_1)$ and second relevant transfer in $(t_1,t_2)$:}
        \begin{align*}
             f_{G_{1d}}(\boldsymbol{h}|\cS,\boldsymbol{t},\lambda) &= 2 \cdot \lambda^2 e^{- \lambda (2t_1 + h_1 + 2 h_2)} \quad \text{with} \quad 0 < h_1 < t_1 \text{ and } 0 < h_2 < t_2-t_1. \\
             f_{G_{2d}}(\boldsymbol{h}|\cS,\boldsymbol{t},\lambda) &= f(G_{3d}, \boldsymbol{h}|\cS,\boldsymbol{t},\lambda) = \lambda^2 e^{- \lambda (2t_1 + h_1 + 2 h_2)}
        \quad \text{with} \quad 0 < h_1 < t_1 \text{ and } 0 < h_2 < t_2-t_1.
        \end{align*}

    \item \textbf{First relevant transfer in $(t_1,t_2)$:}
    \begin{align*}
         f_{G_{1y}}(\boldsymbol{h}|\cS,\boldsymbol{t},\lambda) &=  2 \lambda e^{-\lambda (3t_1 + 2 h_1)} \quad \text{with}  \quad 0 < h_1 < t_2 - t_1.
    \end{align*}
\end{enumerate}
\end{prop}

\begin{proof}
All densities are obtained by combining:
\begin{enumerate}[(i)]
    \item the choice of transfer arc,
    \item exponential waiting-time densities for relevant transfers, and
    \item probabilities of no transfer over specified intervals.
\end{enumerate}
If $X$ denotes the waiting time to the next relevant transfer, then $X \sim \mathrm{Exp}(c\lambda)$, where $c\in\{2,3\}$ is the number of available recipient arcs. Thus,
\[ \mathbb{P}(X>t)=e^{-c\lambda t}.\]

\medskip
We illustrate the calculation for representative cases.
\begin{itemize}
    \item \textbf{One relevant transfer in $(0,t_1)$ and no other relevant transfers:}\\
    Consider gene history $G_{1a}$. The density of its transfer times is the product of the probability of the transfer arc $e_1 \to e_2$, given by $1/6$, the density of the time of the first transfer, given by $3 \lambda e^{-3 \lambda h_1}$ with $0 < h_1 < t_1$, and the probability of no transfer in $(h_1, t_2)$, given by $e^{-2 \lambda (t_2-h_1)}$.  In summary,
    \begin{align*}
    f_{G_{1a}}(\boldsymbol{h}|\cS,\boldsymbol{t},\lambda) &= \frac{1}{6} \cdot 3 \lambda e^{-3 \lambda h_1} \cdot e^{-2 \lambda (t_2-h_1)} = \frac{1}{2} \lambda e^{- \lambda (2 t_2 + h_1)} \quad \text{with} \quad 0 < h_1 < t_1.
    \end{align*}
    Analogous reasoning applies to $G_{1b}$, $G_{2a}$, $G_{2b}$, $G_{3a}$, and $G_{3b}$.

    \item \textbf{Two relevant transfers in $(0,t_1)$:}\\
    Consider gene history $G_{1c}$. The density of its transfer times is the product of the probability of the first transfer arc between $e_1$ and $e_2$, given by $1/6 + 1/6 = 1/3$, the density of the first transfer, given by $3 \lambda e^{- 3 \lambda h_1}$ with $0 < h_1 < t_1$, and the density of the second transfer, given by $2 \lambda e^{-2 \lambda h_2}$ with $0 < h_2 < t_1-h_1$ (notice that the direction of the second transfer is irrelevant in this case). In summary,
       \begin{align*}
         f_{G_{1c}}(\boldsymbol{h}|\cS,\boldsymbol{t},\lambda)
         &= \frac{1}{3} \cdot 3 \lambda e^{-3 \lambda h_1} \cdot 2 \lambda e^{-2 \lambda h_2}
          = 2 \lambda^2 e^{-\lambda (3h_1 + 2h_2)}, \quad 0 < h_1 < t_1 \text{ and } 0 < h_2 < t_1-h_1.
     \end{align*}
     Analogous reasoning applies to $G_{2c}$ and $G_{3c}$.

     \item \textbf{First relevant transfer in $(0,t_1)$ and second relevant transfer in $(t_1,t_2)$:}\\
     First, consider $G_{1d}$. The density of its transfer times is the product of the probability of the first transfer arc $e_1 \leftrightarrow e_2$ (i.e., $e_1 \to e_2$ or $e_2 \to e_1)$, given by $1/6 + 1/6 = 1/3$, the density of the first transfer, given by $3 \lambda e^{-3 \lambda h_1}$ with $0 < h_1 < t_1$, the probability of no transfer in $(h_1,t_1)$, given by $e^{-2 \lambda (t_1-h_1)}$, and the density of the second transfer (whose direction is irrelevant), given by $2 \lambda e^{-2 \lambda h_2}$ with $0 < h_2 < t_2-t_1$. This gives
     \begin{align*}
           f_{G_{1d}}(\boldsymbol{h}|\cS,\boldsymbol{t},\lambda)
         &= \frac{1}{3} \cdot 3 \lambda e^{-3 \lambda h_1} \cdot e^{-2 \lambda (t_1-h_1)} \cdot 2 \lambda e^{-2 \lambda h_2} 
          = 2 \lambda^2 e^{- \lambda (2t_1 + h_1 + 2 h_2)},\\
         &\qquad 0 < h_1 < t_1 \text{ and } 0 < h_2 < t_2-t_1. 
     \end{align*}
     For $G_{2d}$ and $G_{3d}$, the argument is analogous, except that only one direction of the first transfer is feasible (see Appendix~\ref{app:classification}), so the probability of the corresponding transfer arc is $1/6$ rather than $1/3$. Hence,
        \begin{align*}
        f_{G_{2d}}(\boldsymbol{h}\mid \cS,\boldsymbol{t},\lambda)
        = \lambda^2 e^{- \lambda (2t_1 + h_1 + 2 h_2)}, 
        \quad 0 < h_1 < t_1 \text{ and } 0 < h_2 < t_2-t_1,
        \end{align*}
        with an identical expression for $G_{3d}$.

      \item \textbf{First relevant transfer in $(t_1,t_2)$:}\\
      This density is the product of no transfer in $(0,t_1)$, given by $e^{-3 \lambda t_1}$ and the density of the time of the first transfer, given by $2 \lambda e^{-2 \lambda h_1}$ with $0 < h_1 < t_2-t_1$. In this case, the direction of this first transfer arc is irrelevant. In summary,
    \begin{align*}
          f_{G_{1y}}(\boldsymbol{h}|\cS,\boldsymbol{t},\lambda) &= e^{-3 \lambda t_1} \cdot 2 \lambda e^{-2 \lambda h_1} 
         = 2 \lambda e^{-\lambda (3t_1 + 2 h_1)} \quad \text{ with }  \quad 0 < h_1 < t_2 - t_1.
     \end{align*}
\end{itemize}
\end{proof}

\subsection{Gene tree probabilities} \label{app:gene-probs}
The probabilities of the three gene trees $T_1, T_2$, and $T_3$, can be obtained by summing over gene histories and for each history, integrating over the transfer time densities from Proposition~\ref{prop:gene-tree-densities}. This leads to the following proposition.

\begin{prop} \label{prop:gene-tree-probs}
      Let $\cS$ be the $3$-taxon species tree from Fig.~\ref{fig:transfer}(i), and let $\boldsymbol{t} = (t_1,t_2)$ with $t_1 < t_2$ denote the $t$-values of its two interior vertices.
    Under the standard LGT model with transfer rate $\lambda$, the probabilities of the three gene tree topologies $T_1 = ((1,2),3)$, $T_2 = ((1,3),2)$, and $T_3 = ((2,3),1)$ are given by
    \begin{align*}
         \mathbb{P}(T_1 | \mathcal{S},\boldsymbol{t}, \lambda) &= \frac{1}{3} + \frac{2}{3} e^{-3 \lambda t_1},\\
     \mathbb{P}(T_2 | \mathcal{S},\boldsymbol{t}, \lambda) &=   \mathbb{P}(T_3 | \mathcal{S},\boldsymbol{t}, \lambda) = \frac{1}{3} - \frac{1}{3} e^{-3 \lambda t_1}.
    \end{align*}
    In particular, the gene tree topology $T_1$ matching the species tree $\cS$ has the highest probability.
\end{prop}

We note that this result is already established in the literature~\cite[Proposition~9.4]{steel2016}. As the proof in~\cite{steel2016} does not rely on gene history densities, we include an alternative proof below.

\begin{proof}[Proof of Proposition~\ref{prop:gene-tree-probs}]
To obtain the probability of gene tree $T_i$ with $i \in \{1,2,3\}$, we integrate the transfer time densities from Proposition~\ref{prop:gene-tree-densities}, summing over all gene histories whose topologies are isomorphic to that of $T_i$. The no-transfer case contributes a point mass and is included separately, while all other cases are integrated over their transfer times. \medskip
  
    First, for $T_1$, we have
        \begin{align*}
            \mathbb{P}(T_1 | \mathcal{S},\boldsymbol{t}, \lambda) 
            &=   e^{-\lambda (t_1 + 2 t_2)}  
            + 2 \cdot \int\limits_{h_1 = 0}^{t_1} \frac{1}{2} \lambda e^{ - \lambda (2t_2 + h_1)} \, dh_1
            + \int\limits_{h_1 = 0}^{t_1} \int\limits_{h_2=0}^{t_1-h_1} 2 \lambda^2 e^{-\lambda (3h_1+2h_2)} \, dh_2 \, dh_1\\
            &\qquad 
            + \int\limits_{h_1=0}^{t_1} \int\limits_{h_2=0}^{t_2-t_1} 2 \lambda^2 e^{-\lambda (2t_1 + h_1 + 2h_2)} \, dh_2 \, dh_1
            + \int\limits_{h_1=0}^{t_2 - t_1} 2 \lambda e^{- \lambda(3t_1 + 2h_1)} \, dh_1 \\
            &=  e^{-\lambda (t_1 + 2 t_2)} + \left( e^{-2 \lambda t_2} - e^{-\lambda (t_1+2t_2)} \right)
            + \left( \frac{1}{3} + \frac{2}{3} e^{- 3 \lambda t_1} - e^{- 2 \lambda t_1} \right)\\
            &\qquad + \left( - e^{- 3 \lambda t_1} + e^{-2 \lambda t_1} - e^{- 2 \lambda t_2} + e^{- \lambda (t_1+2t_2)} \right) + \left(e^{-3 \lambda t_1} - e^{- \lambda (t_1+2t_2)} \right)\\
            &= \frac{1}{3} + \frac{2}{3} e^{- 3 \lambda t_1}.\\
        \end{align*}

    Similarly, for $T_2$ and $T_3$, \smallskip
    \begin{align*}
         \mathbb{P}(T_2 | \mathcal{S},\boldsymbol{t}, \lambda) &=  \mathbb{P}(T_3 | \mathcal{S},\boldsymbol{t}, \lambda) \\
         &= \int\limits_{h_1=0}^{t_1} \frac{1}{2} \lambda e^{-\lambda (2t_1+h_1)} \, dh_1
         + \int\limits_{h_1=0}^{t_1} \frac{1}{2} \lambda e^{-\lambda (2t_2+h_1)} \, dh_1
         + \int\limits_{h_1=0}^{t_1} \int\limits_{h_2=0}^{t_1-h_1} 2 \lambda^2 e^{-\lambda (3h_1 + 2h_2)} \, dh_2 \, dh_1 \\
         &\qquad + \int\limits_{h_1=0}^{t_1} \int\limits_{h_2=0}^{t_2-t_1} \lambda^2 e^{- \lambda (2t_1 + h_1 + 2h_2)} \, dh_2 \, dh_1 \\
         &= \left( - \frac{1}{2} e^{- 3 \lambda t_1} + \frac{1}{2} e^{- 2 \lambda t_1} \right) + \left( \frac{1}{2} e^{-2 \lambda t_2} - \frac{1}{2} e^{-\lambda (t_1+2t_2)} \right) + \left( \frac{1}{3} + \frac{2}{3} e^{-3 \lambda t_1} - e^{-2 \lambda t_1} \right) \\
         &\qquad + \left( - \frac{1}{2} e^{- 3 \lambda t_1} + \frac{1}{2} e^{- 2 \lambda t_1} - \frac{1}{2} e^{- 2 \lambda t_2} + \frac{1}{2} e^{- \lambda (t_1 + 2t_2)} \right) \\
         &= \frac{1}{3} - \frac{1}{3} e^{-3 \lambda t_1}. \\
    \end{align*}
    
     Clearly, $\mathbb{P}(T_1 | \mathcal{S},\boldsymbol{t}, \lambda) > \mathbb{P}(T_2 | \mathcal{S},\boldsymbol{t}, \lambda) =  \mathbb{P}(T_3 | \mathcal{S},\boldsymbol{t}, \lambda)$, which completes the proof.
\end{proof}

\subsection{Site pattern probability derivations} \label{app:site-pattern-probabilities}

Consider the gene history $G_{1x}$ and let $\boldsymbol{g} = (g_1,g_2)$ with $g_1 < g_2$ denote the $t$-values of its two interior vertices. Define
\begin{align}
    d_0 &= \frac{1}{4} + \frac{3}{4} e^{- g_1 \mu}, \quad d_1 = \frac{1}{4} - \frac{1}{4} e^{- g_1 \mu}, \quad 
    b_0 = \frac{1}{4} + \frac{3}{4} e^{\mu(g_2-g_1)}, \quad \text{and} \quad b_1 = \frac{1}{4} - \frac{1}{4}e^{- (g_2-g_1) \mu}. \label{def:d0d1b0b1}
\end{align}
Then, we have (see also~\cite{garcia2007}):
\begin{align*}
    \mathbb{P}(xxx| G_{1x}, \boldsymbol{g}) &= b_0^2 d_0^3 + 3 b_0^2 d_1^3 + 6 b_0 b_1 d_0^2 d_1 + 6 b_0 b_1 d_0 d_1^2 + 12 b_0 b_1 d_1^3 + 3 b_1^2 d_0^3 + 6 b_1^2 d_0^2 d_1 + 6 b_1^2 d_0 d_1^2 + 21 b_1^2 d_1^3; \\
    \mathbb{P}(xxy|  G_{1x}, \boldsymbol{g}) &= 6 b_0^2 d_0^2 d_1+6 b_0^2 d_0 d_1^2+12 b_0^2 d_1^3+12 b_0 b_1 d_0^2 d_1+84 b_0 b_1 d_0 d_1^2+48 b_0 b_1 d_1^3+30 b_1^2 d_0^2 d_1 \\
    &\qquad +102 b_1^2 d_0 d_1^2+84 b_1^2 d_1^3;\\
    \mathbb{P}(xyx|  G_{1x}, \boldsymbol{g}) &= \frac{1}{2} \left(3 b_0^2 d_0^2 d_1+3 b_0^2 d_0 d_1^2+6 b_0^2 d_1^3+6 b_0 b_1 d_0^3+12 b_0 b_1 d_0^2 d_1+12 b_0 b_1 d_0 d_1^2+42 b_0 b_1 d_1^3 \right. \\
    &\left. \qquad +6 b_1^2 d_0^3+21 b_1^2 d_0^2 d_1+21 b_1^2 d_0 d_1^2+60 b_1^2 d_1^3\right);\\
    \mathbb{P}(xyy|  G_{1x}, \boldsymbol{g}) &= \mathbb{P}(xyx|  G_{1x}, \boldsymbol{g});  \\
    \mathbb{P}(xyz|  G_{1x}, \boldsymbol{g}) &= 18 b_0^2 d_0 d_1^2+6 b_0^2 d_1^3+24 b_0 b_1 d_0^2 d_1+60 b_0 b_1 d_0 d_1^2+60 b_0 b_1 d_1^3+24 b_1^2 d_0^2 d_1+114 b_1^2 d_0 d_1^2+78 b_1^2 d_1^3.
\end{align*}

Now, the site pattern probabilities for any of the other gene histories are obtained analogously, accounting for symmetry and adjusting the $g$-values in \eqref{def:d0d1b0b1} according to the gene history classification in Appendix~\ref{app:classification}. For example, for $G_{1a}$, we have $g_1 = h_1$ and $g_2 = t_2$. 

By averaging the conditional site pattern probabilities over the LGT distribution of gene histories and transfer times (Proposition~\ref{prop:gene-tree-densities}), we obtain the marginal site pattern probabilities $\boldsymbol{\overline{p}} = (p_{xxx}, p_{xxy}, p_{xyx}, p_{yxx}, p_{xyz})$. 

For example, let $\mathcal{G}$ denote the set of all gene histories classified in Appendix~\ref{app:classification}. Then
\begin{align*}
p_{xxx}
&= \mathbb{P}(xxx| G_{1x}, \boldsymbol{g}) \cdot e^{-\lambda(t_1+2t_2)} + \sum_{G \in \mathcal{G} \setminus \{G_{1x}\}} \iint \mathbb{P}(xxx | G, \boldsymbol{g}) \,
f_{G}(\boldsymbol{h} | \cS, \boldsymbol{t}, \lambda)\, d\boldsymbol{h}.
\end{align*}

Here, $\cS$ is the $3$-taxon species tree from Fig.~\ref{fig:transfer}(i), with interior times $\boldsymbol{t} = (t_1, t_2)$ satisfying $t_1 < t_2$. For a given gene history $G$, $\boldsymbol{g} = (g_1,g_2)$ denotes the interior vertex times of $G$, and $\boldsymbol{h} = (h_1,h_2)$ the times of the relevant transfers (when present). The limits of integration are as specified in Proposition~\ref{prop:gene-tree-densities}.

Note that the no-transfer case $G_{1x}$ contributes a point mass, and is therefore treated separately, while for all other gene histories we integrate over their corresponding transfer-time densities.

\medskip
Performing all calculations with the computer algebra system \texttt{Mathematica}~\cite{mathematica}, we obtain the following expressions for the marginal site pattern probabilities.

\begin{prop}\label{prop:site-pattern-probabilities} 
Let $\cS$ be the $3$-taxon species tree from Fig.~\ref{fig:transfer}(i), and let $\boldsymbol{t} = (t_1,t_2)$ with $t_1 < t_2$ denote the $t$-values of its two interior vertices. 
Under the standard LGT model with transfer rate $\lambda$, the marginal site pattern probabilities under the JC69 model with $\mu = 4/3$ are given by
\allowdisplaybreaks
\begin{align*}
    p_{xxx} &= \frac{(9 \lambda +2) (3 \lambda  (9 \lambda +10)+4)}{(3 \lambda +4) (9 \lambda +4) (9 \lambda +8)}+\frac{3}{8} e^{-\frac{1}{3} (3 \lambda +4) (3 t_1 +2 t_2)} 
    \left(\frac{3 \lambda  e^{2 \lambda  t_1+\frac{20 t_1}{3}}}{3 \lambda +4}+\frac{3 (3 \lambda -4) \lambda  e^{\frac{2}{3} (3 \lambda +8) t_1}}{(3 \lambda +4)^2} \right. \\
    & \left. \qquad -\frac{12 (3 \lambda +2) \lambda  e^{3 \lambda t_1+\frac{20 t_1}{3}}}{(3 \lambda +4)^2}-e^{2 (\lambda +2) t_1}+\frac{(4-3 \lambda ) e^{2 \lambda  t_1+\frac{8 t_1}{3}}}{3 \lambda +4}+\frac{4 (3 \lambda +2) e^{(3 \lambda +4) t_1}}{3 \lambda +4} \right. \\
    & \left. \qquad +\frac{3 (3 \lambda +16) \lambda  e^{4 t_1+2 \lambda  t_2+\frac{8 t_2}{3}}}{(3 \lambda -8) (3 \lambda +4)}+\frac{3 (3 \lambda  (3 \lambda  (9 \lambda +68)+80)-64) \lambda  e^{\frac{8 (t_1+t_2)}{3}+2 \lambda  t_2}}{(3 \lambda -4) (3 \lambda +4)^2 (9 \lambda +4)} \right.\\
    & \left. \qquad +\frac{12 (3 \lambda +2) \lambda  e^{(\lambda +4) t_1+2 \lambda  t_2+\frac{8 t_2}{3}}}{(3 \lambda +4)^2}-\frac{4 (3 \lambda  (3 \lambda  (3 \lambda -4)-4)-32) e^{\frac{1}{3} (3 \lambda +4) (t_1+2 t_2)}}{(3 \lambda -8) (3 \lambda -4) (3 \lambda +4)} \right. \\
    & \left. \qquad -\frac{(9 \lambda +4) e^{\frac{4}{3} (t_1+2 t_2)+2 \lambda  t_2}}{9 \lambda +8}-\frac{3 \lambda  e^{2 \lambda  t_2+\frac{8 t_2}{3}}}{3 \lambda +4}\right); \\[2em]
    p_{xxy} &= \frac{12 (\lambda  (9 \lambda +10)+2)}{(3 \lambda +4) (9 \lambda +4) (9 \lambda +8)}+\frac{3}{16} e^{-\frac{1}{3} (3 \lambda +4) (3 t_1+2 t_2)} \left(-\frac{6 (\lambda -2) e^{(3 \lambda +4) t_1}}{3 \lambda +4}+\frac{18 \lambda  (\lambda -2) e^{3 \lambda  t_1+\frac{20 t_1}{3}}}{(3 \lambda +4)^2} \right. \\
    &  \left. \qquad +e^{2 (\lambda +2) t_1}-\frac{3 \lambda  e^{2 \lambda  t_1+\frac{20 t_1}{3}}}{3 \lambda +4}+\frac{(3 \lambda -16) e^{2 \lambda  t_1+\frac{8 t_1}{3}}}{3 \lambda +4}-\frac{3 \lambda  (3 \lambda -16) e^{\frac{2}{3} (3 \lambda +8) t_1}}{(3 \lambda +4)^2} \right.\\
     & \left. \qquad +\frac{2 \left(3 \lambda  \left(9 \lambda ^2-30 \lambda -8\right)-64\right) e^{\frac{1}{3} (3 \lambda +4) (t_1+2 t_2)}}{(3 \lambda -8) (3 \lambda -4) (3 \lambda +4)}-\frac{3 \lambda  \left(3 \lambda  \left(81 \lambda ^2+456 \lambda +80\right)-256\right) e^{\frac{8 (t_1+t_2)}{3}+2 \lambda  t_2}}{(3 \lambda -4) (3 \lambda +4)^2 (9 \lambda +4)} \right.\\
     & \left. \qquad -\frac{18 \lambda  (\lambda -2) e^{(\lambda +4) t_1+2 \lambda  t_2+\frac{8 t_2}{3}}}{(3 \lambda +4)^2}+\frac{2 (5 \lambda  (3 \lambda -2)-16) e^{4 t_1+2 \lambda  t_2+\frac{8 t_2}{3}}}{(3 \lambda -8) (3 \lambda +4)}-\frac{4 (3 \lambda +2) e^{\frac{4}{3} (t_1+2 t_2)+2 \lambda  t_2}}{9 \lambda +8} \right.\\
     &\left. \qquad +\frac{3 \lambda  e^{2 \lambda  t_2+\frac{8 t_2}{3}}}{3 \lambda +4}\right);\\[2em]
     p_{xyx} &= p_{xyy} 
     = \frac{12 (\lambda  (9 \lambda +10)+2)}{(3 \lambda +4) (9 \lambda +4) (9 \lambda +8)}+\frac{3}{32} e^{-\frac{1}{3} (3 \lambda +4) (3 t_1+2 t_2)} \left(\frac{6 (9 \lambda -2) \lambda  e^{3 \lambda  t_1+\frac{20 t_1}{3}}}{(3 \lambda +4)^2}+\frac{9 \lambda  e^{2 \lambda  t_1+\frac{20 t_1}{3}}}{3 \lambda +4} \right.\\
     &\left. \qquad -\frac{3 (15 \lambda -8) \lambda  e^{\frac{2}{3} (3 \lambda +8) t_1}}{(3 \lambda +4)^2}-3 e^{2 (\lambda +2) t_1}+\frac{(15 \lambda -8) e^{2 \lambda  t_1+\frac{8 t_1}{3}}}{3 \lambda +4}-\frac{2 (9 \lambda -2) e^{(3 \lambda +4) t_1}}{3 \lambda +4} \right.\\
     &\left. \qquad -\frac{3 \left(3 \lambda  \left(81 \lambda ^2+768 \lambda +400\right)-128\right) \lambda  e^{\frac{8 (t_1+t_2)}{3}+2 \lambda  t_2}}{(3 \lambda -4) (3 \lambda +4)^2 (9 \lambda +4)}-\frac{6 (9 \lambda -2) \lambda  e^{(\lambda +4) t_1+2 \lambda  t_2+\frac{8 t_2}{3}}}{(3 \lambda +4)^2} \right. \\
     &\left. \qquad -\frac{6 \lambda  e^{\frac{4}{3} (t_1+2 t_2)+2 \lambda  t_2}}{9 \lambda +8}+\frac{6 (\lambda  (3 \lambda  (9 \lambda -22)-104)+64) e^{\frac{1}{3} (3 \lambda +4) (t_1+2 t_2)}}{(3 \lambda -8) (3 \lambda -4) (3 \lambda +4)}-\frac{4 (\lambda  (3 \lambda -29)-8) e^{4 t_1+2 \lambda  t_2+\frac{8 t_2}{3}}}{(3 \lambda -8) (3 \lambda +4)} \right. \\
     &\left. \qquad +\frac{15 \lambda  e^{2 \lambda  t_2+\frac{8 t_2}{3}}}{3 \lambda +4}\right);\\[2em]
     & \\
     p_{xyz} &= \frac{24 (3 \lambda +2)}{(3 \lambda +4) (9 \lambda +4) (9 \lambda +8)}+\frac{3}{4} e^{-\frac{1}{3} (3 \lambda +4) (3 t_1+2 t_2)} \left(-\frac{3 \lambda  e^{2 \lambda  t_1+\frac{20 t_1}{3}}}{3 \lambda +4}+\frac{24 \lambda  e^{3 \lambda  t_1+\frac{20 t_1}{3}}}{(3 \lambda +4)^2} \right. \\
     &\left. \qquad +\frac{3 (3 \lambda -4) \lambda  e^{\frac{2}{3} (3 \lambda +8) t_1}}{(3 \lambda +4)^2}+e^{2 (\lambda +2) t_1}+\frac{(4-3 \lambda ) e^{2 \lambda  t_1+\frac{8 t_1}{3}}}{3 \lambda +4}-\frac{8 e^{(3 \lambda +4) t_1}}{3 \lambda +4} \right. \\
     &\left. \qquad +\frac{3 (3 \lambda  (3 \lambda  (9 \lambda +68)+80)-64) \lambda  e^{\frac{8 (t_1+t_2)}{3}+2 \lambda  t_2}}{(3 \lambda -4) (3 \lambda +4)^2 (9 \lambda +4)}-\frac{3 (3 \lambda +16) \lambda  e^{4 t_1+2 \lambda  t_2+\frac{8 t_2}{3}}}{(3 \lambda -8) (3 \lambda +4)} \right. \\
     &\left. \qquad -\frac{24 \lambda  e^{(\lambda +4) t_1+2 \lambda  t_2+\frac{8 t_2}{3}}}{(3 \lambda +4)^2}+\frac{8 (3 \lambda -2) (3 \lambda +8) e^{\frac{1}{3} (3 \lambda +4) (t_1+2 t_2)}}{(3 \lambda -8) (3 \lambda -4) (3 \lambda +4)}+\frac{(9 \lambda +4) e^{\frac{4}{3} (t_1+2 t_2)+2 \lambda  t_2}}{9 \lambda +8}-\frac{3 \lambda  e^{2 \lambda  t_2+\frac{8 t_2}{3}}}{3 \lambda +4}\right).\\
\end{align*}
\end{prop}

\begin{proof}
    We provide the \texttt{Mathematica}~\cite{mathematica} file containing all calculations as a supplementary file to this manuscript.
\end{proof}

\subsection{Additional simulation results}\label{appendix:simulation}
Fig. \ref{fig:sim1-results-0.8} shows the results of the first simulation study in the case that $\lambda = 0.8$.

\begin{figure}[htbp]
    \centering
    \includegraphics[scale=0.6]{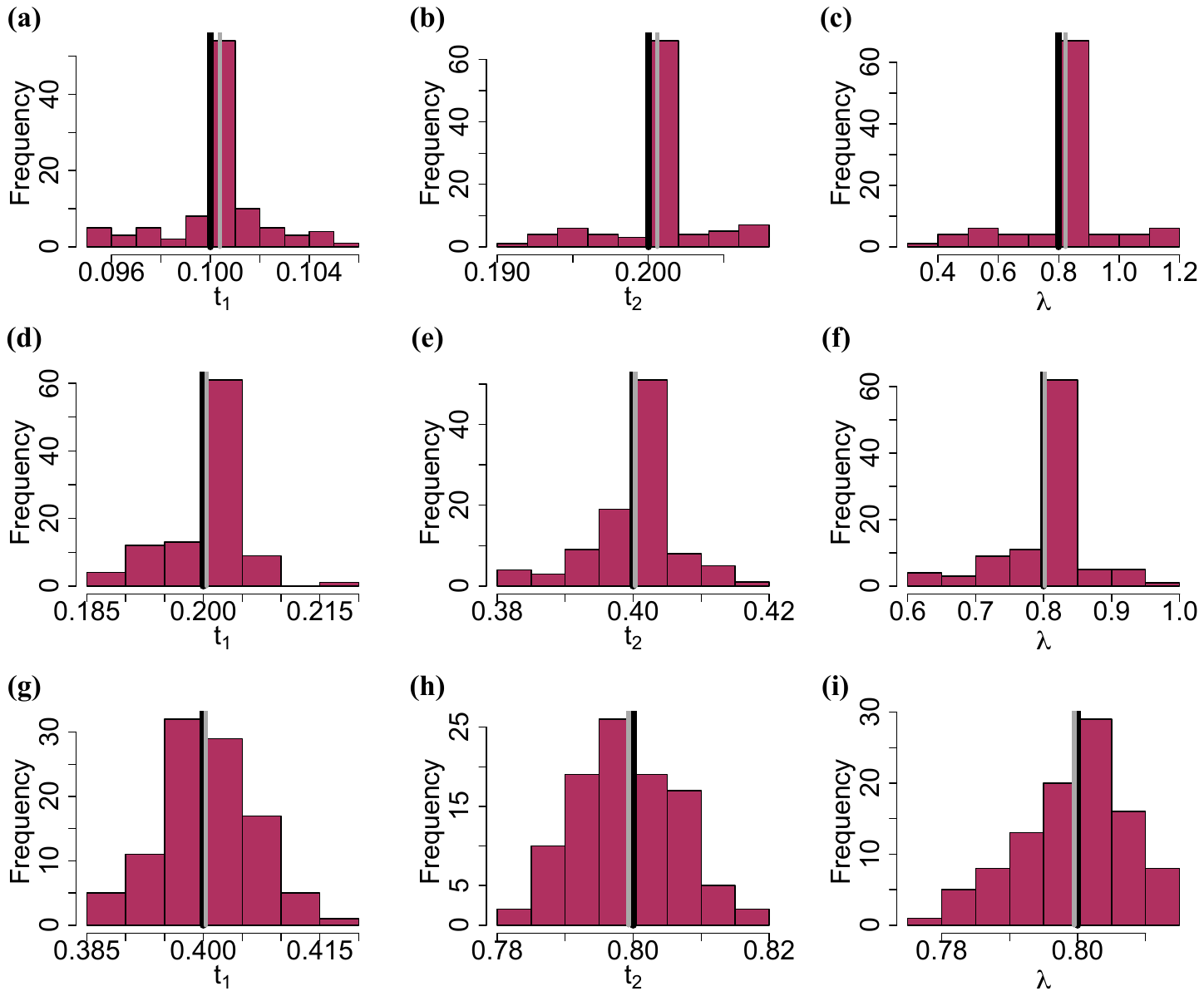}
    \caption{Results of the first simulation study when $\lambda=0.8$.  (a), (b), (c) The first row shows histograms of the 100 MLEs for parameters $t_1$, $t_2$, and $\lambda$, respectively, when the true values are $t_1=0.1$ and $t_2=0.2$. (d), (e), (f)  The second row shows histograms of the 100 MLEs for parameters $t_1$, $t_2$, and $\lambda$, respectively, when the true values are $t_1=0.2$ and $t_2=0.4$. (g), (h), (i) The third row shows histograms of the 100 MLEs for parameters $t_1$, $t_2$, and $\lambda$, respectively, when the true values are $t_1=0.4$ and $t_2=0.8$. In all plots, the black vertical line denotes the true parameter value and the gray vertical line gives the average estimate over the 100 replicates. }
    \label{fig:sim1-results-0.8}
\end{figure}

\end{document}